\documentclass[twoside,11pt]{article}

\usepackage{jmlr2e}
\usepackage{amssymb, amsmath, mathtools} 
\usepackage{graphicx, float} 
\usepackage{chngcntr} 
\usepackage{hyperref} 
\usepackage{bbm} 
\usepackage{indentfirst} 
\usepackage{microtype} 
\usepackage{xcolor} 

\usepackage{thmtools} 
\usepackage{enumerate}

\graphicspath{ {figures/} }

\counterwithin{equation}{section}
\counterwithin{figure}{section}

\renewcommand{\P}{\mathbb{P}} 
\newcommand{\Var}{\mathrm{Var}} 
\newcommand{\Cov}{\mathrm{Cov}} 
\newcommand{\1}{\mathbbm{1}} 

\newcommand{\R}{\mathbb{R}} 
\renewcommand{\O}{\mathcal{O}} 
\newcommand{\iid}{\stackrel{i.i.d.}{\sim}} 
\DeclareMathOperator*{\argmin}{arg\,min} 
\DeclareMathOperator{\Tr}{Tr} 

\DeclarePairedDelimiter\floor{\lfloor}{\rfloor} 
\DeclarePairedDelimiter\norm{\lVert}{\rVert} 

\ShortHeadings{Spectral Deconfounding}{\'{C}evid, B\"{u}hlmann and Meinshausen}
\firstpageno{1}

\begin{document}

\title{Spectral Deconfounding via Perturbed Sparse Linear Models}

\author{\name Domagoj \'{C}evid \email cevid@stat.math.ethz.ch \\
       \addr Seminar f\"{u}r Statistik\\
       ETH Z\"{u}rich\\
       8092 Z\"{u}rich, Switzerland
       \AND
       \name Peter B\"{u}hlmann \email peter.buehlmann@stat.math.ethz.ch \\
       \addr Seminar f\"{u}r Statistik\\
       ETH Z\"{u}rich\\
       8092 Z\"{u}rich, Switzerland
       \AND
       \name Nicolai Meinshausen \email meinshausen@stat.math.ethz.ch \\
       \addr Seminar f\"{u}r Statistik\\
       ETH Z\"{u}rich\\
       8092 Z\"{u}rich, Switzerland
       }

\editor{}

\maketitle

\begin{abstract}
Standard high-dimensional regression methods assume that the underlying coefficient vector is sparse. This might not be true in some cases, in particular in presence of hidden, confounding variables. Such hidden confounding can be represented as a high-dimensional linear model where the sparse coefficient vector is perturbed. For this model, we develop and investigate a class of methods that are based on running the Lasso on preprocessed data. The preprocessing step consists of applying certain spectral transformations that change the singular values of the design matrix. We show that, under some assumptions, one can achieve the usual Lasso $\ell_1$-error rate for estimating the underlying sparse coefficient vector, despite the presence of confounding. Our theory also covers the Lava estimator \citep{lava2017} for a special model class. The performance of the methodology is illustrated on simulated data and a genomic dataset. 
\end{abstract}



\section{Introduction}
Many datasets nowadays include measurements from many variables. The corresponding models are typically high-dimensional with many more parameters than the sample size.
For statistical estimation and inference, there is a vast literature which assumes sparsity. For example, see the monographs by \citet{buhlmann2011statistics}, \citet{giraud2014introduction} or \citet{hastie2015statistical}. 

However, the performance of many high-dimensional regression methods might suffer in presence of unobserved confounding variables which affect both the predictors and the response. Confounding is a severe issue when interpreting regression parameters, often, but not necessarily, in connection with causal inference. A prime example are genetic studies where unobserved confounding can easily lead to spurious correlations and partial dependencies \citep{novembre2008genes}. Even when one is concerned with only prediction, the causal parameter leads to predictive robustness against perturbations of the confounding variables.

Adjusting for unobserved confounding variables is very important in practice and several deconfounding methods have been suggested for various settings \citep{gerard2017empirical, leek2007capturing, gagnon2012using, wang2018blessings, paul2008preconditioning}. Often, the methods try to estimate the confounding variables directly from the data, usually by using some factor analysis technique. There are not many theoretical results justifying the methods, especially since some of them are quite complicated and therefore difficult to analyze.

Our focus is on linear models. In absence of confounding variables, when the response is affected only by a small number of predictors, i.e. the coefficient vector is sparse, one can efficiently estimate the active set and the corresponding coefficients with the Lasso and related methods and thus achieve the minimax optimal $\ell_1$-norm estimation error rate, see, for example, \citet{bickel2009simultaneous} or the monographs by \citet{buhlmann2011statistics} or \citet{wainwright2019high}. However, these methods are not adequate in presence of confounding in linear model, since in addition to just a few predictors that indeed affect the response, we have additional association of the response with many other predictors, as they contain information about the confounding variables.

Some approaches for relaxing the sparsity assumption are (i) the notion of weak sparsity \citep{van2016estimation}, where the regression parameter $\beta$ fulfills the condition that $\norm{\beta}_q$ is small for some $0 < q < 1$ or (ii) assuming the structure that the regression parameter can be represented as a sum of a sparse and a dense vector. The case (i) does not call for a new method or algorithm: in fact, the Lasso still exhibits optimal convergence rate if $\norm{\beta}_q$ is sufficiently small \citep{van2016estimation}. On the other hand, case (ii) requires a different method such as, for example, Lava \citep{lava2017}. 

Here we investigate how to deal with the confounding by analyzing the second case where the parameter is a sum of a sparse and a dense part. If many predictors are affected by the confounding variables, the true underlying regression vector will be changed by some small, dense perturbation. We propose left multiplying the response $Y$ and the design matrix $X$ consisting of the values of the predictors by a carefully chosen spectral transformation matrix $F$ which transforms the singular values of $X$. The transformed response and design matrix can then be used as the input for a high-dimensional sparse regression technique: we consider the Lasso as a prime example. We investigate the theoretical properties and empirical performances for the class of spectral transformations. As a result, we conclude that certain spectral transformations that shrink the large singular values, such as the Trim transform which we introduce in this paper, perform well over a range of scenarios, pointing out also some advantages over other techniques and approaches.

\subsection{Relation to other work and our contribution}
For adjusting for the effect of unobserved confounding, the most prominent method in practice is to adjust for the top several principal components of the predictors, see for example \citep{novembre2008genes}. Such PCA adjustment is also a special case of the FarmSelect estimator \citep{fan2020factor} for the linear model, which considers the problem of high-dimensional variable selection where the latent variables cause the correlations of the predictors, but do not directly affect the response. PCA adjustment is a special case of a spectral transformation. Our presented theory explains when and why this method works well and proposes an alternative transformation, called Trim transform, which has an advantage that one does not need to estimate the number of principal components to adjust for.

The Puffer transform, which maps all singular values to $1$, has also been suggested for improving the variable selection properties of the Lasso for a sparse high-dimensional linear model \citep{jia2015preconditioning}. Our theory gives a more precise result about the Puffer transform for the estimation problem: the Trim transform is at least as good as Puffer transform and substantially better when the sample size is close to the number of predictors. In \citet{shah2018rsvp}, the Puffer transform in combination with bootstrap aggregation is used in order to estimate the covariance matrix in presence of confounding variables, a very different quantity than the precision matrix or regression coefficients.

\citet{chandrasekaran2012latent} address the problem of estimating the precision matrix in presence of a few hidden confounding variables. Then the observed precision matrix can be represented as a sum of the initial sparse precision matrix and a low-rank perturbation due to the confounding variables. Their model is similar to the one we consider, but the assumptions and the goals differ. We aim to estimate just the regression coefficients instead of the whole precision matrix and the method we propose is much simpler. Furthermore, the theoretical conclusions are substantially different: we establish the convergence rates in terms of the $\ell_1$-norm estimation error, while they consider support recovery and $\ell_\infty$-norm bounds for the low-dimensional setting, assuming strong conditions. Also \citet{fan2013large} have considered low rank plus sparse problems from the viewpoint of factor models: their contribution provides a rich source of references from an area which is vaguely related to our current work.

The Lava estimator \citep{lava2017} is the most similar to our Trim transform. The theory we develop, covering also the Lava, gives a result for the $\ell_1$-norm estimation error rate for the sparse coefficient vector. This goes well beyond the theory given by \citet{lava2017} for justifying the original and interesting Lava method. There, the authors mostly consider the Gaussian sequence model but also provide general bounds for high-dimensional regression whose (e.g. asymptotic) behavior is not further analyzed in terms of restricted eigenvalues and the sparse and dense component of the underlying unknown parameter vector. Our presented theory exploits the specific structure of a hidden confounding model which provides a different motivation than the one in \citet{lava2017}, where no confounding was considered. In addition, our developments suggest a simple rule for the choice of the $\ell_2$-norm regularization parameter for the Lava estimator, leaving only the $\ell_1$-norm regularization parameter as the single parameter to be tuned by cross-validation.

Our contribution can be seen as threefold. We describe a class of spectral transformations and propose a simple spectral transformation called Trim transform, which is perhaps slightly easier to use than the Lava or the PCA adjustment estimator. Furthermore, for the linear model where the underlying sparse parameter has been perturbed, we provide novel theory establishing for a certain class of spectral transformations a fast convergence rate for the $\ell_1$-norm estimation error of the true underlying sparse parameter. Finally, and as our primary goal, we use these results to show how the issue of hidden confounding can be addressed by using a wisely chosen spectral transformation, such as e.g. Trim transform, with the Lasso afterwards: we establish under certain assumptions the same convergence rate as the one of the Lasso for a linear model without confounding and illustrate the empirical performance of our method on simulated and real genomic data. Our method is entirely modular and can be used not only in conjunction with the Lasso, but also any other reasonable high-dimensional linear regression method.

\section{The models} \label{model}
In this section we consider a linear model with additional confounding. We also introduce a perturbed linear model and show how it relates to the confounding model. Our theoretical results apply to the perturbed linear model as well and it is useful for better understanding of the confounding model.

\subsection{Confounding model}
Consider a standard (high-dimensional) linear model with $n$ observations and $p$ predictors $X_1, \ldots, X_p$ linearly affecting the response $Y$.
Suppose further that $q$ additional unobserved confounding variables linearly affect the response as well. The confounding variables are correlated with the predictors, introducing additional spurious correlations between the response and the predictors.

The model for $n$ i.i.d. observations is given by:
\begin{align} \label{confounding_model}
Y &= X\beta + H\delta + \nu
\end{align}
where $X \in \R^{n \times p}$ is the matrix of predictors and $H \in \R^{n \times q}$ represents the hidden confounding variables, which exhibit correlation with $X$, i.e., $\Cov(H,X) \neq 0$ (with a slight abuse of notation, we write $\Cov(H,\,X)$ as the covariance of any row of $H$ and $X$). We assume that $X$ and $H$ have i.i.d. rows that are jointly Gaussian and that $\nu \in \R^n$ is a vector of sub-Gaussian errors with mean zero and standard deviation $\sigma_\nu$, independent of $X$ and $H$. The vectors $\beta \in \R^p$ and $\delta \in \R^q$ are fixed coefficients; we additionally assume that $\beta$ is sparse with exactly $s$ non-zero components. Since the model does not change under the transformation $H \leftarrow H \Cov(H)^{-1/2}$, $\delta \leftarrow \Cov(H)^{1/2}\delta$, we can assume without loss of generality that $\Cov(H) = I_q$, i.e. the confounding variables are uncorrelated. 

Note that by $L_2$ projection, $X$ can also be written as 
\begin{equation}\label{add-PB1}
X = H\Gamma + E,
\end{equation}
where we choose $\Gamma \in \R^{q \times p}$ such that $\Cov(H, E) = 0$: 
$$\Gamma = \Cov(H)^{-1}\Cov(H, X) = \Cov(H, X).$$

The matrix $\Gamma \in \R^{q \times p}$ describes the linear effect of confounding variables on $X$. The random term $E \in \R^{n \times p}$ can be seen as the unconfounded design matrix; without confounding, i.e. when $H = 0$, it equals $X$. The columns of $E$ are allowed to be correlated and we denote its covariance matrix by $\Sigma_E$; if the components of $E$ are (weakly) uncorrelated, $X$ is generated from an (approximate) factor model \citep{anderson1958introduction, chamberlain1982arbitrage}. Here the hidden variables do not encode a factor structure for $X$ alone, but also in addition generate confounding effects.

A main example of the above model is a structural equation model (SEM)
\begin{eqnarray*}
& &X \leftarrow H\Gamma + E,\\
& &Y \leftarrow X\beta + H\delta + \eta
\end{eqnarray*}
and thus $\beta$ is the direct causal effect of $X$ on $Y$. In a standard SEM with no further hidden variables, the components of $E$ would be assumed independent.

We will show in Section \ref{theory} that one can recover the coefficient $\beta$ if the confounding is dense in a certain sense, e.g. when the rows or columns of $\Gamma = \Cov(H, X)$ are realizations of independent and identically distributed random variables with mean zero.

\subsection{Perturbed linear model}
The confounding model (\ref{confounding_model}) is related to the perturbed linear model
\begin{equation} \label{perturbed_model}
Y = X(\beta + b) + \epsilon,
\end{equation}
where the sparse coefficient vector $\beta$ has been perturbed by the perturbation vector $b \in \R^p$ and $\epsilon \in \R^n$ is the vector of sub-Gaussian errors independent of $X$ with standard deviation $\sigma$. Here we assume that the rows of $X$ are i.i.d. sub-Gaussian vectors with mean zero and covariance matrix $\Sigma = \Cov(X)$.

The relationship between models arises by rewriting (\ref{confounding_model}) as
$$Y = X(\beta + b) + (H\delta - Xb) + \nu,$$
where $b$ satisfies that $\Cov(X, H\delta - Xb) = 0$, i.e., $Xb$ is the $L_2$-projection of $H \delta$ onto $X$. This gives us the formula
\begin{align} \label{perturbation_formula}
b &= \Cov(X)^{-1}\Cov(X, H)\delta \nonumber \\
&= \left(\Cov(X, H) \Cov(H)^{-1} \Cov(H, X) + \Cov(E)\right)^{-1}\Cov(X, H)\delta
\end{align}

The error is given by $\epsilon = (H\delta - Xb) + \nu$, which by construction of $b$ is uncorrelated with $X$ and thus independent of $X$, because the rows of $X$ and $H$ are assumed to be jointly Gaussian in the confounding model. We require such independence (induced by joint Gaussianity) in the proof of Theorem \ref{general_bound}, although $\epsilon$ being uncorrelated with $X$ might be sufficient. The variance of the error is given by
$$\sigma^2 = \Var(H\delta - Xb + \nu) \leq \norm{\delta}_2^2 + \sigma_\nu^2.$$

One can think of $H\delta - Xb$ as the part of the confounding that can not be explained by $X$ and which just increases the variance of the additive error. $Xb$ is the part of the confounding effect $H\delta$ that is correlated with $X$ and, as is well known, the bias $b$ due to the confounding makes the estimation of $\beta$ more difficult.

In conclusion, the confounding model (\ref{confounding_model}) can be thought of as a special case of the perturbed linear model (\ref{perturbed_model}), but with additional relationship between the design matrix $X$, the perturbation vector $b$, given by (\ref{perturbation_formula}), and the additive error $\epsilon$.

The perturbed linear model is in general unidentifiable since we can only infer $\beta + b$ from the data generating distribution. This makes the estimation of $\beta$ impossible, unless $b$ has a certain structure; we will be able to asymptotically retrieve the sparse coefficient vector $\beta$, by assuming, for example, that $b$ converges to $0$ in some norm. In Section \ref{theory}, we investigate under which conditions we are able to infer the sparse part $\beta$ and how efficiently in terms of statistical accuracy. 

It could be interesting to estimate the coefficient vector $\beta + b$ rather than just $\beta$, but it is impossible to do in general in the high-dimensional case; even if we knew $\beta$ exactly, estimating $b$ would mean estimating $p$ coefficients from $n < p$ data points, which is impossible without additional assumptions about the structure of $b$.

\subsection{Relationship with the factor model literature}
Even though the confounding variables are hidden, we are able to infer some of their properties if they affect many of the observed predictors $X$. This is the essence of factor analysis, where a lot of interesting work has been done. If the latent factors $H$ linearly affect the covariates, as it is the case in the confounding model \eqref{confounding_model}, they can be estimated well (up to a rotation) from the principal components of the design matrix $X = H\Gamma + E$ \citep{chamberlain1982arbitrage, bai2003inferential}, especially if one additionally imposes certain assumptions on the factor loadings $\Gamma$ \citep{bing2017adaptive}.

There are several related models considered in the literature. In certain cases \citep{paul2008preconditioning, bing2019essential} we assume that only the latent factors affect the response and the observed covariates are only used to obtain information about the latent factors:
$$Y = H\delta + \nu, \qquad X = H\Gamma + E.$$
In \citet{bai2006confidence} one has an additional contribution of some other known low-dimensional covariates $W$:
$$Y = W\beta + H\delta + \nu, \qquad X = H\Gamma + E.$$
Another line of work assumes that the latent factors do not directly affect the response:
$$Y = X\beta + \nu, \qquad X = H\Gamma + E,$$
but that they only cause the predictors to be correlated \citep{huang2011variable, fan2020factor}. Such correlation makes the analysis much more difficult, especially for the problem of variable selection, and one can use the factor analysis to address this issue.

In this paper we allow the latent confounders to affect both the predictors and the response and focus on the estimation of the sparse coefficient vector $\beta$, which has a causal interpretation as it describes the direct effect of the predictors on the response. The key difficulty is to handle the bias $b$ in the observational data caused by the latent confounders.
The assumption of dense confounding, expressed in detail in Section \ref{theory}, is related to the spiked covariance assumptions common in the factor analysis literature \citep{paul2008preconditioning, bai2003inferential}. It is used to make conclusions about the structure of the coefficient perturbation $b$ rather than about the factor identifiability. We avoid estimating the factor variables directly, but instead we adjust for them implicitly, by transforming the singular values of $X$.

\section{Methodology}
In the following, we propose and motivate some methods based on a class of spectral transformations.

\subsection{Spectral transformations}
Let $X = UDV^T$ be the singular value decomposition of $X$, where $U \in \R^{n \times r}, D \in \R^{r \times r}, V \in \R^{p \times r}$,  where $r = \min(n, p)$ is the rank of $X$. We write $d_1 \geq d_2 \geq \ldots \geq d_r$ for the diagonal elements of $D$. We use the truncated form of SVD, which uses only non-zero singular values.

The idea is to first transform our data by applying some specific linear transformation $F:\R^n \to \R^n$ and then perform the Lasso algorithm: 
\begin{align} \label{transformed_lasso}
X &\to \tilde{X} \coloneqq FX \nonumber\\
Y &\to \tilde{Y} \coloneqq FY \nonumber \\
\hat{\beta} = \argmin_\beta &\left\{ \frac{1}{n}\|\tilde{Y} - \tilde{X}\beta\|_2^2 + \lambda\|\beta\|_1 \right\}.
\end{align}

We restrict our attention to the class of spectral transformations, which transform the singular values of $X$, while keeping its singular vectors intact. Let $\tilde{D}$ be an arbitrary $r \times r$ diagonal matrix with diagonal elements $\tilde{d}_1, \ldots, \tilde{d}_r$. Our spectral transformation matrix is given by
\begin{equation}
F = U \begin{bmatrix}\tilde{d}_1/ d_1 & 0 & \ldots & 0\\ 0 & \tilde{d}_2/ d_2  & \ldots & 0\\ \vdots & \vdots & \ddots & \vdots\\ 0 & 0 & \ldots & \tilde{d}_r/d_r\\\end{bmatrix}U^T \label{spectraltransformation}
\end{equation}
and then we have 
\begin{equation*}
\tilde{X} = FX = U\tilde{D}V^T
\end{equation*}

In this paper we explore the question of what is a good choice of $F$ for the estimation of $\beta$. In general, the Lasso performs best when the predictors are uncorrelated and when the errors are independent. Therefore, a good choice of $F$ needs to find a good balance between a well behaved error term $\tilde{\epsilon} = F\epsilon$, well behaved design matrix $\tilde{X}$ and well behaved perturbation term $\tilde{X}b$.

One such transformation is the \textbf{Trim transform} which limits all singular values to be at most some constant $\tau$:
\begin{equation}\label{trim}
\tilde{d}_i = \min(d_i, \tau).
\end{equation}
We show in Section \ref{theory} that it can, under some assumptions, achieve the same $\ell_1$-norm error rate for the estimation of the unknown sparse coefficient vector $\beta$ as the Lasso in the case of no confounding. We also show that the median singular value is a good choice of $\tau$:
$$\tau = d_{\floor{r/2}}$$

\subsection{Existing methods and motivation}\label{motivation}

We discuss some existing methods which are related to the spectral transformation method described above and provide further explanations and relationships between them. We also present intuitive explanation why our suggested method should work well against dense confounding.

\subsubsection{Examples of spectral transformations}
Several existing methods consist of first transforming the data with a certain matrix $F$ (some of which fall into class of spectral transformations (\ref{spectraltransformation})), and then using some regression method, such as the Lasso.
\paragraph{Lava}
One such example is the Lava estimator \citep{lava2017}, designed for the linear model where the coefficient vector can be written as a sum of a dense and a sparse vector. It is originally given by (with a slight change of notation)
\begin{equation*}
(\hat{\beta},\, \hat{b}) = \argmin_{\beta, b} \left\{ \frac{1}{n} \|Y - X(\beta + b)\|_2^2 + \lambda_2\|b\|_2^2 + \lambda_1\|\beta\|_1 \right\},
\end{equation*}
which can be seen as a combination of Lasso and Ridge regression. It is shown in \citet{lava2017} that the solution of this optimization problem is given by
\begin{align*}
F &= (I_p - X(X^TX+n\lambda_2I_p)^{-1}X^T)^{1/2} ,\\
\hat{\beta} &= \argmin_{\beta} \left\{ \frac{1}{n} \|\tilde{Y} - \tilde{X}\beta\|_2^2 + \lambda_1\|\beta\|_1 \right\} ,\\
\hat{b} &= (X^TX+n\lambda_2I_p)^{-1}X^T(Y-X\hat{\beta}) .
\end{align*}
From here, one can see that the estimator of the sparse part is just a Lasso estimator applied to the transformed data, where
\begin{equation*}
\tilde{d}_i = \sqrt{\frac{n\lambda_2 d_i^2}{n\lambda_2 + d_i^2}} .
\end{equation*}
This transformation is visualized in Figure \ref{transformations}.

\paragraph{Puffer transform}
Another example is the Puffer transform introduced in \citet{jia2015preconditioning}, which uses the Lasso after mapping all non-zero singular values $d_i$ to a constant $\tilde{d}_i = 1$. The algorithm is analyzed as a preconditioning method for the variable selection problem without any coefficient perturbation. This transformation decreases the correlations between the columns of the design matrix, but it can inflate the errors, especially when $p$ is close to $n$. It can also be thought of as a special case of the Lava transformation in the case when $\lambda_2 \to 0$, since then $\tfrac{\widetilde{d}_i}{\sqrt{n\lambda_2}} \to 1$ (the denominator here is just a scaling factor). The transformation is displayed in Figure \ref{transformations}.

\paragraph{PCA adjustment}
Another example of a spectral transformation is given by PCA-based methods for adjusting for hidden confounders \citep{novembre2008interpreting, fan2020factor, bai2003inferential}. In the confounding model \eqref{confounding_model}, the effect of confounding variables will approximately lie in the span of the first few principal components of $X$ (see Figure \ref{projection}). One adjusts for a first few principal components from the columns of the design matrix $X$ before further analysis in hope of removing the effect of the confounding variables \citep{paul2008preconditioning, huang2011variable}. This procedure is in fact analogous to applying a spectral transformation, where the matrix $\tilde{D}$ is obtained from $D$ by mapping the first several singular values to $0$. See also Figure \ref{transformations} for an illustration. The slight difficulty with this approach is knowing exactly the number of principal components to remove. Asymptotically, this can be done with high probability \citep{bai2003inferential} under certain assumptions on the separation of the singular values. However, for finite samples or if there is a slight model misspecification, it might not be that easy to estimate $q$, see e.g. our real data genomic dataset in Figure \ref{expression-svd}.

\subsubsection{Some intuition}
Since our method (\ref{transformed_lasso}) is invariant under transformation $F \to cF$, for arbitrary constant $c \in \R$, we can assume without loss of generality that the singular values of $F$ are at most $1$, i.e. the transformation $F$ shrinks all vectors, with different shrinkage in directions of its singular vectors. Ideally, we would like to shrink in a way such that the perturbation term $\tilde{X}b$ becomes much smaller compared to the signal $\tilde{X}\beta$. 

Trim transform has the highest shrinkage along directions of the singular vectors corresponding to large singular values. The more $b$ is aligned with the first few singular vectors of $X$ (those corresponding to large singular values), the larger $\norm{Xb}_2$ will be. Therefore, shrinking those large singular values ensures that $\norm{\tilde{X}b}$ stays small regardless of the direction $b$ is pointing to. It is especially the case in the confounding model that $b$ approximately lies in the span of the first few singular vectors (see Figure \ref{projection}).

\begin{figure}[h!]
\centering
\hspace*{0cm}
\includegraphics[scale=0.3]{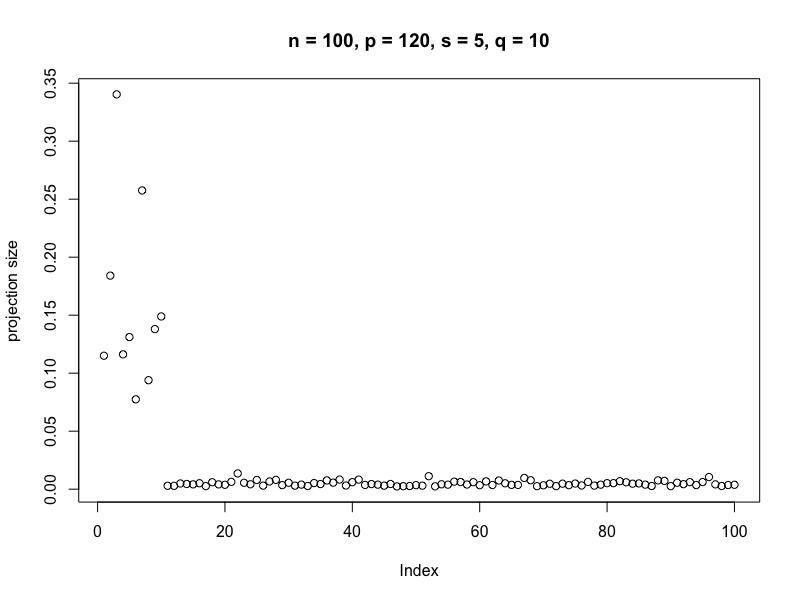}
\caption{Size of the projection of $b$ onto $V_i$ for different $i$, for a random dataset drawn from the confounding model with $q=10$ confounding variables, as described in Section \ref{simsetting}. We see that the projections of $b$ onto the first $10$ singular values are substantially larger than the rest.}
\label{projection}
\end{figure}

As can be seen from definition of $b$, $Xb$ is the part of the confounding effect $H\delta$ which is correlated with $X$. Therefore, $\norm{Xb}_2$ can be just as large as $\norm{H\delta}_2 = \O(\sqrt{n}\norm{\delta}_2)$. However, after applying the Trim transform we have that \[\norm{\tilde{X}b}_2 \leq \lambda_{\max}(\tilde{X})\|b\|_2 = \O\left(\sqrt{p} \times \sqrt{\frac{\norm{\delta}_2^2}{p}}\right) = \O(\norm{\delta}_2),\] which is substantially smaller than before. $\lambda_{\max}(\tilde{X})$ is the largest singular value of $\tilde{X}$, which will be shown in Lemma \ref{averagesingularvalue} to be of order $\sqrt{p}$ for the Trim transform and we have $\norm{b}_2 = \O\big(\sqrt{\norm{\delta}_2^2/p}\big)$ under certain model assumptions by Lemma \ref{perturbationsize}.

On the other hand, the signal $X\beta$ lies in the span of a sparse set of predictors. Therefore, the signal $\tilde{X}\beta$ will be approximately of the same size as the signal $X\beta$ before transformation, unless $\beta$ is aligned with the large singular vectors, which are shrunk the most. This is very unlikely if they are sufficiently random. This is illustrated in Figure \ref{scatter}. Therefore, by shrinking large singular values, $\norm{Xb}_2$ will decrease much more compared to $\norm{X\beta}_2$.

\begin{figure}[h!]
\centering
\includegraphics[scale=0.45]{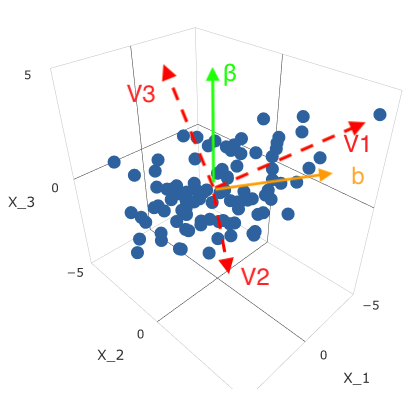}
\caption{Visualisation of the relationship between the perturbation $b$, signal $\beta$ and singular vectors of $X$. In the confounding model $b$ will be much more aligned with the singular vectors corresponding to large singular values than $\beta$.}
\label{scatter}
\end{figure}

\section{Theoretical Results} \label{theory}
In this section we analyse the behaviour of the $\ell_1$-estimation error for the sparse coefficient $\beta$ for an arbitrary spectral transformation $F$. We derive results for the perturbed linear model (\ref{perturbed_model}) and relate them to the confounding model (\ref{confounding_model}) by using the relationship between them.

We show that if our spectral transformation fulfils certain criteria, and the confounding is dense in the sense that every confounding variable affects many predictors, we achieve in the high-dimensional case the same $\ell_1$-error rate as the Lasso in the case when we have no confounding, despite the presence of the coefficient perturbation caused by the confounding variables. Furthermore, in Section \ref{justification}, we discuss specific choices of spectral transformations and verify that the Trim transform (\ref{trim}), as well as Lava and PCA adjustment, can be used in order to achieve this error rate.

We assume first for simplicity that we are in the high-dimensional case, where $p \geq n$. However, the theory developed in this section also holds for the case $n > p$ with small adjustments. We discuss the case $n > p$ in more details in Section \ref{lowdim}.

\subsection{Notation}
$$\phi_M \coloneqq \inf_{\|\alpha\|_1 \leq 5\|\alpha_S\|_1} \frac{\sqrt{\alpha^TM\alpha}}{\frac{1}{\sqrt{s}}\|\alpha_S\|_1} ,$$
where $S$ is the support set of $\beta$, $s$ is the size of $S$ and $\alpha_S$ is a vector consisting only of the components of $\alpha$ which are in $S$.

Let us also write $\tilde{\Sigma} \coloneqq \frac{1}{n}\tilde{X}^T\tilde{X}$, and $\hat{\Sigma} \coloneqq \frac{1}{n}X^TX$.
We denote the $k$-th largest diagonal element of the transformed singular values $\tilde{D}$ by $\tilde{d}_{(k)}$. We denote the the largest, the smallest and $i$-th (non-zero) singular value of any rectangular matrix $A$ by $\lambda_{\max}(A)$, $\lambda_{\min}(A)$ and $\lambda_i(A)$ respectively. The condition number is defined as $\text{cond}(A) = \tfrac{\lambda_{max}(A)}{\lambda_{min}(A)}$.

Finally, we use the notation $A = \Omega(B)$ if $\frac{B}{A} = \O(1)$, i.e.\ if $A$ has asymptotically at least the same rate as $B$ and $A \asymp B$ if $A$ and $B$ have asymptotically the same rate. $A = \O_p(B)$ means that there exists a constant $c > 0$ such that $\P(A > cB) \to 0$ and $\Omega_p$ is defined analogously. 

\subsection{Main result for the confounding model}
We present here the main result for the confounding model (\ref{confounding_model}), which we derive below by considering the relationship with the corresponding perturbed linear model, as described in Section \ref{model}.

\begin{restatable}{theorem}{optimalrate}\label{optimalrate}
Consider the model in (\ref{confounding_model}) with $\max_{i} \Sigma_{ii} = \O(1)$ and $\text{cond}(\Sigma_E) = \O(1)$ and suppose that $\lambda_{\min}(\Sigma)$ is bounded away from zero. Assume that the model satisfies
\begin{itemize}
\setlength{\leftskip}{1cm}
\item[\textbf{(A1)}] $\lambda_{\min}(\Gamma) = \lambda_{\min}(\Cov(X, H)) = \Omega(\sqrt{p})$. 
\end{itemize}
\setlength{\leftskip}{0cm}
Assume additionally that a spectral transformation $F$ in (\ref{transformed_lasso}) with 
$\lambda_{\max}(F) = 1$ satisfies
\begin{itemize}
\setlength{\leftskip}{1cm}
\item[\textbf{(A2)}] $\lambda_{\max}(\tilde{X}) = \O_p(\sqrt{p})$
\item[\textbf{(A3)}] $\phi_{\tilde{\Sigma}}^2 = \Omega_p(\lambda_{\min}(\Sigma))$.
\end{itemize}
\setlength{\leftskip}{0cm}
Then for the penalty level $\lambda \asymp \sigma\sqrt{\frac{\log p}{n}}$, despite the confounding variables, the $\ell_1$-estimation error has the following rate:
$$\|\hat{\beta} - \beta\|_1  = \O_p\left(\frac{\sigma s}{\lambda_{\min}(\Sigma)}\sqrt{\frac{\log p}{n}}\right).$$
\end{restatable}

The assumption \textbf{(A1)} means that the confounding is dense in the sense that each confounding variable is correlated with many predictors:
The condition $\lambda_{\min}(\Gamma) = \Omega_p(\sqrt{p})$ is satisfied, for example, if $\frac{q}{p} \to 0$ and $\Gamma$ is drawn at random with either rows or columns of $\Gamma$ being independent, identically distributed sub-Gaussian random vectors, as shown in Lemma \ref{perturbationsize}.

We also show in Section \ref{justification} that certain choices of the spectral transformation, such as the Trim transform (\ref{trim}) with $\tau = d_{\floor{tn}}$, where $t \in (0, 1)$ is an arbitrary constant, or the PCA adjustment, which maps first several singular values to zero, satisfy with high probability the conditions \textbf{(A2)} and \textbf{(A3)} in the high-dimensional setting under certain conditions.

\subsection{\texorpdfstring{$\ell_1$}{l1}-estimation error of $\beta$ in the perturbed linear model} \label{bounds}
In this section we derive an upper bound for the $\ell_1$-estimation error of $\beta$ in the perturbed linear model and show that we can achieve the usual Lasso error rate in the high-dimensional case, provided the perturbation $b$ is sufficiently small. Then the main theorem for the confounding model, Theorem \ref{optimalrate}, follows from Corollary \ref{optimal_rate_perturbed} by using the relationship between the models described in Section \ref{model}.

The following result describes the effect of an arbitrary linear transformation $F$ on the $\ell_1$-estimation error of the Lasso:
\begin{restatable}{theorem}{generalbound} \label{general_bound}
Assume the model in (\ref{perturbed_model}) with $\max_{i} \Sigma_{ii} = \O(1)$. Let $F \in \R^{n \times n}$ be an arbitrary linear transformation and $A>0$ an arbitrary fixed constant. Then for the method described in $(\ref{transformed_lasso})$ with transformation $F$ and penalty level $\lambda = A\sigma\sqrt{\frac{\log p}{n}}\lambda_{\max}(F)^2$, with probability at least $1 - 2p^{1-A^2/(32\max_i \Sigma_{ii})} - pe^{-n/136}$, we have
$$\|\hat{\beta} - \beta\|_1 \leq C_1\frac{s\lambda}{\phi_{\tilde{\Sigma}}^2}  + C_2\frac{\|\tilde{X}b\|_2^2}{n\lambda} ,$$
where $C_1, C_2$ are constants depending only on $A$.
\end{restatable}
\begin{remark}
One can get a better bound 
$$\|\hat{\beta} - \beta\|_1 \leq C_1\frac{s\lambda}{\phi_{\tilde{\Sigma}}^2}  + C_2\frac{\sqrt{s}}{\phi_{\tilde{\Sigma}}}\frac{\norm{\tilde{X}b}_2}{\sqrt{n}}$$
by taking larger penalty $\lambda$ than the one above, but then $\lambda$ depends on the unknown quantity $\norm{\tilde{X}b}_2$. For that reason we will use the bound above with standard penalty level $\lambda$, since it does not matter when $\norm{\tilde{X}b}_2$ is small, which holds in our case, as shown later.
\end{remark}

The first term is the standard bound for the $\ell_1$-error of the Lasso, with only difference that the compatibility constant is for the matrix $\tilde{\Sigma} = \frac{\tilde{X}^T\tilde{X}}{n}$ rather than the matrix $\hat{\Sigma} = \frac{X^TX}{n}$. The second term shows the dependence of the error on the term $\tilde{X}b$. It is also worth noting that the penalty $\lambda$ has standard form up to the scaling correction factor $\lambda_{\max}(F)^2$, which equals $1$ for the Trim transform and the PCA adjustment.

In order to control the error caused by the coefficient perturbation $b$, we need to make $\norm{\tilde{X}b}_2$ small by shrinking the singular values enough, e.g. by ensuring that $\tilde{d}_{(1)}$, the largest singular value after transformation, is small. On the other hand, we must not shrink the singular values too much, since we need $\phi_{\tilde{\Sigma}}$ to stay large. If we have that $\phi^2_{\tilde{\Sigma}}$ is bounded away from $0$ with high probability, as it is the case with $\phi^2_{\hat{\Sigma}}$ (see \citet{buhlmann2011statistics}), and that $\norm{\tilde{X}b}_2$ is sufficiently small, we get from Theorem \ref{general_bound} that our estimator achieves the usual Lasso error rate:

\begin{restatable}{corollary}{optimalrateperturbed}\label{optimal_rate_perturbed}
Consider the model in (\ref{perturbed_model}) with $\max_{i} \Sigma_{ii} = \O(1)$ and suppose that $\lambda_{\min}(\Sigma)$ is bounded away from zero. For the coefficient perturbation $b$ as in (\ref{perturbation_formula}), assume that
\begin{itemize}
\setlength{\leftskip}{1cm}
\item[\textbf{(A1')}] $\|b\|_2^2 = \O\left(\frac{s\sigma^2 \log p}{p}\right)$.
\end{itemize}
\setlength{\leftskip}{0cm}
Assume additionally that the spectral transformation $F$ in (\ref{transformed_lasso}) with 
$\lambda_{\max}(F) = 1$ satisfies
\begin{itemize}
\setlength{\leftskip}{1cm}
\item[\textbf{(A2)}] $\lambda_{\max}(\tilde{X}) = \O_p(\sqrt{p})$
\item[\textbf{(A3)}] $\phi_{\tilde{\Sigma}}^2 = \Omega_p(\lambda_{\min}(\Sigma)).$
\end{itemize}
\setlength{\leftskip}{0cm}
Then for the penalty level $\lambda \asymp \sigma\sqrt{\frac{\log p}{n}}$, despite the coefficient perturbation, the $\ell_1$-estimation error has the following rate:
$$\|\hat{\beta} - \beta\|_1  = \O_p\left(\frac{\sigma s}{\lambda_{\min}(\Sigma)}\sqrt{\frac{\log p}{n}}\right).$$
\end{restatable}

We show in the following section that in the perturbed linear model that arises from the confounding model (\ref{confounding_model}), the induced coefficient perturbation $b$, given in (\ref{perturbation_formula}), satisfies the condition \textbf{(A1')}, provided that the dense confounding assumption \textbf{(A1)} is satisfied.
We also show that certain spectral transformations, such as the Trim transform (\ref{trim}) with $\tau = d_{\floor{tn}}$, where $t \in (0, 1)$ is an arbitrary constant, or the PCA adjustment satisfy the conditions \textbf{(A2)} and \textbf{(A3)} under certain conditions.

\begin{remark} [Fixed design]
The results of Theorem \ref{general_bound} and Corollary \ref{optimal_rate_perturbed} can be easily extended to the perturbed linear model with fixed design. One can even relax the assumption \textbf{(A1')} to a weaker condition \[\norm{V^Tb}_2^2 = \O\left(\frac{s\sigma^2\log p}{p}\right).\]
It is worth noting that if the perturbation vector $b$ has uniformly random direction, which is not the case with the confounding model (\ref{confounding_model}), this becomes much weaker than the condition \textbf{(A1')} above and we only require $\norm{b}_2^2 = \O\left(\frac{s\sigma^2\log p}{n}\right)$.
\end{remark}

\subsection{Validity of the assumptions} \label{justification}
In this section we will justify the assumptions in Theorem \ref{optimalrate} and Corollary \ref{optimal_rate_perturbed} for certain spectral transformations $F$, with an emphasis on the Trim transform (\ref{trim}) and the PCA adjustment. We also discuss later the performance of other choices of spectral transformations.

\subsection*{Assumptions \textbf{(A1)} and \textbf{(A1')}} 
The assumption \textbf{(A1')} for the perturbed linear model says that the coefficient perturbation must not be too large. It can also be viewed as the condition which makes the perturbed linear model identifiable, since in general it is impossible to distinguish the true coefficient vector $\beta$ from the perturbed coefficient vector $\beta + b$, unless $b$ has some additional structure. The rate $\O(\sqrt{ s \sigma^2 \log p / p})$ may seem too strict, but this is the rate with respect to the $\ell_2$-norm, so if the perturbation vector is dense, this becomes approximately $\|b\|_1 = \O(\sqrt{s \sigma^2 \log p})$.

The following lemma shows that if the confounding is dense in the confounding model (the assumption \textbf{(A1)} holds), then the induced coefficient perturbation in the underlying perturbed linear model is small (the assumption \textbf{(A1')} holds). It is important to note that certain dense confounding assumption is necessary. The term $Xb$ can be thought of as the part of the confounding $H\delta$ that can be explained by $X$ and if, as an extreme example, the confounder $H_i$ is correlated with only the predictor $X_j$, only the $j$-th component of $X$ will be useful for describing the effect of $H_i$ on $Y$ and thus $b_j$ will be very large and we will not be able to estimate $\beta_j$.

\begin{restatable}{lemma}{perturbationsize} \label{perturbationsize}
Assume that the confounding model (\ref{confounding_model}) satisfies $\lambda_{\min}(\Gamma) = \lambda_{\min}(\Cov(H, X)) = \Omega\left(\sqrt{p}\right)$ and $\text{cond}(\Sigma_E) = \O(1)$. Then we have:
$$\|b\|_2^2 = \|\Cov(X)^{-1}\Cov(X, H)\delta\|_2^2 \leq \text{cond}(\Sigma_E)\cdot\frac{\norm{\delta}_2^2}{\lambda_{\min}(\Gamma)^2} = \O\left(\frac{\norm{\delta}_2^2}{p}\right) = \O\left(\frac{\sigma^2}{p}\right)$$
The condition $\lambda_{\min}(\Gamma) = \Omega_p(\sqrt{p})$ is satisfied, for example, if $\frac{q}{p} \to 0$ and $\Gamma$ is drawn at random with either its rows or columns being independent, identically distributed sub-Gaussian random variables with expectation $0$ and covariance matrix $\Sigma_\Gamma$, with $\lambda_{\min}(\Sigma_\Gamma)$ bounded away from zero.
\end{restatable}

From this we see that it is important that the effect of the latent variables is spread out over many predictors. If this is not true, $\lambda_{\min}(\Gamma)$ will be too small and thus $\norm{b}_2$ will be too large.

\subsection*{Assumption \textbf{(A2)}}
We investigate quickly the behaviour of singular values of $X$ in order to see whether the assumption \textbf{(A2)} holds for the transformed matrix $\tilde{X}$. This assumption says that after the transformation, the largest singular value is not too large.

In the confounding model we have $\Sigma = \Gamma^T\Gamma + \Sigma_E$, i.e. the covariance matrix of $X$ has additional low-rank component $\Gamma^T\Gamma$, which causes the top several singular values of $\Sigma$ to be very large. Since the rows of $X$ are drawn from a distribution with covariance matrix $\Sigma$, the first few singular values of $X$ will be large as well \citep{donoho2013optimal}. However, the following lemma shows that the bulk of the singular values will never be too large, i.e. they will be of order $\sqrt{p}$. The assumption \textbf{(A2)} requires the transformed singular values to be of this order. 

\begin{restatable}{lemma}{averagesingularvalue} \label{averagesingularvalue}
Assume that $X \in \R^{n\times p}$ is a random matrix whose rows are i.i.d. sub-Gaussian vectors with covariance matrix $\Sigma$. Let $d_1,\ldots, d_r \geq 0$ be its singular values. Assume also that $\Tr(\Sigma) \asymp p$ and that $\sqrt{\log p /n} \to 0$. We have:
$$\frac{1}{n}\sum_{i=1}^r d_i^2 = \Tr(\Sigma)(1 + o_p(1)).$$
Furthermore, when $p > n$, $d_{\floor{tn}} = \O_p(\sqrt{p})$ for any $t \in (0, 1)$.
\end{restatable}

For the Trim transform the largest singular value after transformation $\tilde{d}_{(1)}$ equals the trimming threshold $\tau$ and the above lemma shows that $\tau = d_{\floor{tn}}$ for $t \in (0, 1)$, e.g. the median singular value when $t=0.5$, is a good choice and the assumption \textbf{(A2)} holds.

If we further assume that $\Sigma_E$ has bounded singular values, thus ensuring the gap between the $q$-th and $(q+1)$-st eigenvalues of $\Sigma$, we get that all but the first $q$ singular values of $X$ will not be too large, thus justifying the assumption \textbf{(A2)} for the PCA adjustment, since there we have $\lambda_{\max}(\tilde{X}) = \tilde{d}_{(1)} = \tilde{d}_{q+1} = d_{q+1}$.

\begin{restatable}{lemma}{PCAsingularvalue} \label{PCAsingularvalue}
Assume that $p > n$ and that $X$ has i.i.d. sub-Gaussian rows with covariance matrix $\Sigma = \Gamma^T\Gamma + \Sigma_E$, where $\Gamma \in \R^{q\times p}$ and $\lambda_{\max}(\Sigma_E) = \O(1)$, then we have $d_{q+1} = \O_p(\sqrt{p})$.
\end{restatable}

This lemma also shows that in this case the trimming threshold $\tau$ for the Trim transform can be chosen to be $\tau = d_{q+1}$, but $\tau = d_{\floor{tn}}$ might be a better choice as the number of confounders $q$ is unknown.

\subsection*{Assumption \textbf{(A3)}}
This assumption says that the compatibility constant $\phi_{\hat{\Sigma}}$ does not substantially decrease after applying our transformation $F$.
We want to show that by shrinking the singular values we have not shrunk our signal $X\beta$ too much. Intuitively, this means that the active set $X_S$ is not too aligned with the directions along which we substantially shrink, which corresponds to the first several singular vectors in the case of Trim transform and PCA adjustment.

It is difficult to bound $\phi_{\tilde{\Sigma}}$ for an arbitrary spectral transformation $F$, since the distribution of the singular vectors $V$ of the design matrix $X$ is complicated. However, one can directly exploit the results from the factor analysis literature \citep{bai2003inferential} for the PCA adjustment, from which it follows that in a certain asymptotic regime the transformed design matrix $\tilde{X}$ is close to the unconfounded design matrix $E$. Using this result, one can directly obtain the compatibility condition \textbf{(A3)} for the PCA adjustment by using the standard argument \citep{buhlmann2011statistics}.

\begin{restatable}{lemma}{PCAcompatibility} \label{PCAcompatibility}
Let $X$ be generated from the confounding model \eqref{confounding_model} and let $F$ be a spectral transformation shrinking the first $q$ singular values of $X$ to $0$. If $q$ is fixed, $\tfrac{1}{p}\sum_{i,j=1}^p |(\Sigma_E)_{ij}|$ upper bounded and $\tfrac{s p\log p}{n \min(n, p)} \to 0$, we have that, with probability converging to $1$, the compatibility condition holds for the transformed design matrix $\tilde{X} = FX$:
$$\phi_{\tfrac{1}{n}\tilde{X}^T\tilde{X}}^2 \overset{p}{\to} \phi_{\tfrac{1}{n}E^TE}^2 = \Omega_p\left(\lambda_{\min}(\Sigma_E)\right).$$
\end{restatable}

In the Appendix A.1 the analysis of the compatibility constant $\phi_{\tilde{\Sigma}}$ is provided under the somewhat restrictive assumption that the singular vectors $V$ have uniformly distributed direction, but allowing for a more high-dimensional asymptotic regime than in the Lemma \ref{PCAcompatibility}. 

Since the ratio of the transformed singular values for the Trim transform and PCA adjustment is bounded from below by $\tfrac{\tau}{d_{q+1}}$, the compatibility constant $\phi_{\text{Trim}}$ for the Trim transform can be bounded from below by the compatibility constant $\phi_{\text{PCA}}$ for the PCA adjustment:
$$\phi_{\text{Trim}} \geq \tfrac{\tau}{d_{q+1}}\phi_{\text{PCA}} = \tfrac{d_{\floor{tn}}}{d_{q+1}}\phi_{\text{PCA}}$$
and thus the compatibility condition holds for the Trim transform as well if $d_{q+1}$ and $\tau = d_{\floor{tn}}$ are of comparable sizes, i.e. $\tfrac{d_{\floor{tn}}}{d_{q+1}} = \Omega_p(1)$. 
By Lemma \ref{PCAsingularvalue}, we have $d_{q+1} = \O_p(\sqrt{p})$ and by the following lemma it holds that for quite a wide range of settings we also have that $d_{\floor{tn}}^2 = \Omega_p(\lambda_{\min}(\Sigma)p)$. Therefore, Lemma \ref{PCAcompatibility} can be used for showing the compatibility condition for the Trim transform as well.

\begin{restatable}{lemma}{singularvalues} \label{singularvalues}
Assume that $X$ is a random design matrix with i.i.d. rows with covariance matrix $\Sigma$ and suppose $p > n$. Assume that any of the following conditions is satisfied:
\begin{enumerate}[i)]
    \item the rows of $X$ have a sub-Gaussian distribution and $\frac{p}{n} \to \infty$
    \item the rows of $X$ have a $N(0, \Sigma)$ distribution and $\liminf \frac{p}{n} > 1$ 
    \item the rows of $X$ have $N(0, \Sigma)$ distribution and $\limsup \frac{k}{n} < 1$
\end{enumerate}
Then we have
$$d_k^2 = \Omega_p(\lambda_{\min}(\Sigma) p).$$
\end{restatable}

\subsection*{Performance of various spectral transformations} \label{spectral}
The result of Theorem \ref{optimalrate} can be applied to any spectral transformation that satisfies the assumptions (\textbf{A2}) and (\textbf{A3}). We discuss here which spectral transformations satisfy them and what are their possible advantages and disadvantages for the performance of the corresponding estimator $\hat{\beta}$. The illustration of the spectral transformations discussed below is given in Figure \ref{transformations}.

\begin{figure}[h!]
\centering
\hspace*{-0.5cm}
\includegraphics[scale=0.18]{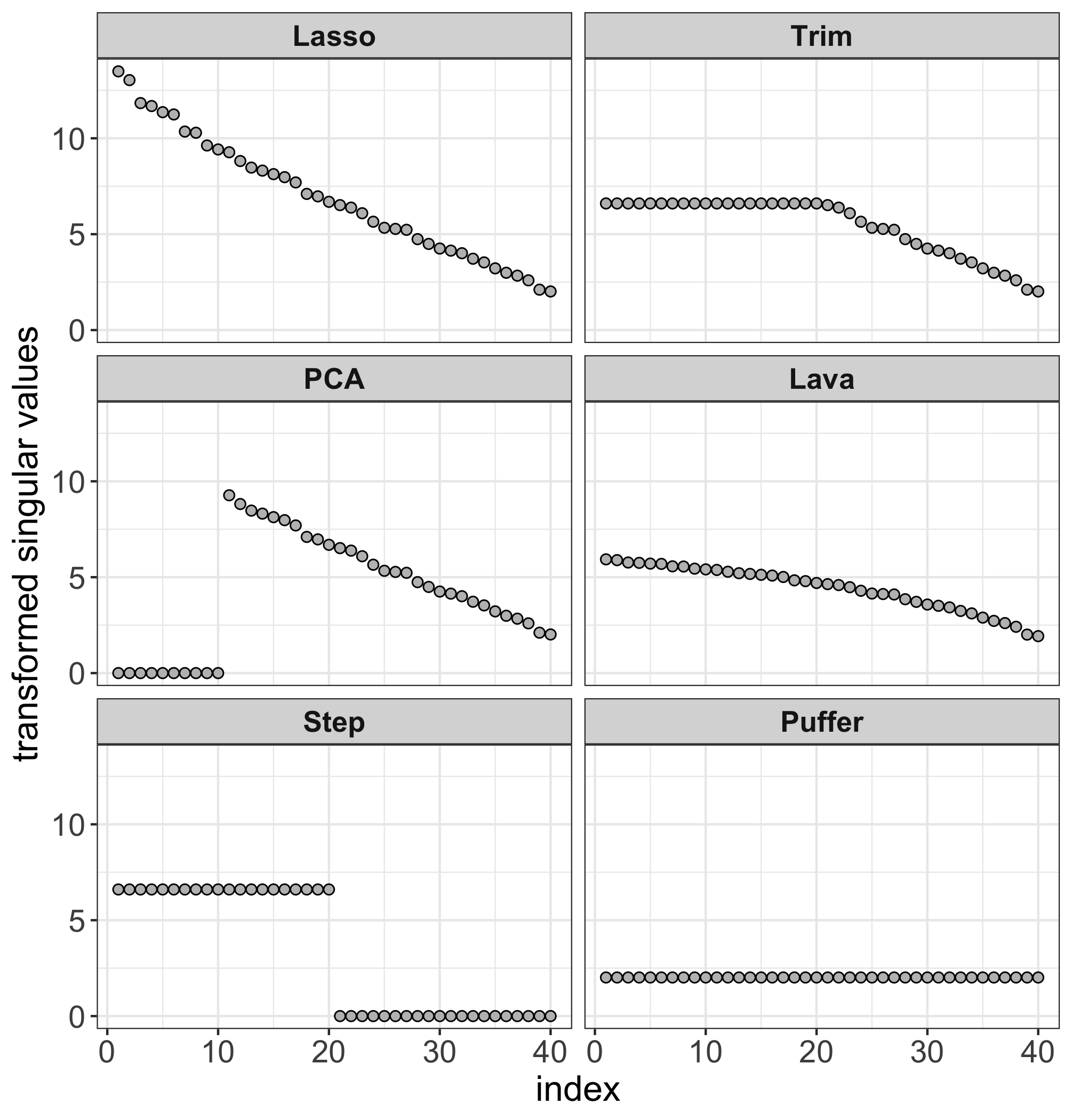}
\caption{Singular values of $\tilde{X}$ after applying spectral transformations corresponding to different methods to $40 \times 60$ matrix $X$ with i.i.d.\ standard normal entries.}
\label{transformations}
\end{figure}

\paragraph{PCA adjustment} As shown above, under certain assumptions we get that the spectral transformation which maps first $q$ singular values to $0$ will satisfy assumptions \textbf{(A2)} and \textbf{(A3)}. Even though it might seem that one disadvantage of this method is that the number of confounding variables $q$ needs to be estimated from the data, one can show that asymptotically it can be done accurately with high probability \citep{bai2003inferential}. PCA adjustment leaves most of the singular values intact, so the increase in the estimator variance will not be large.

\paragraph{Lasso} The simplest option is to take $\tilde{d}_i = d_i$, i.e.\ the usual Lasso algorithm without any transformation. Standard Lasso theory shows that the assumption \textbf{(A3)} is satisfied (see \citet{buhlmann2011statistics}). However, \textbf{(A2)} requires that the largest singular value of $X$ is of order $\O(\sqrt{p})$, which typically does not hold in presence of confounding variables.

\paragraph{Trim transform} As shown above, we have that the Trim transform satisfies assumptions \textbf{(A2)} and \textbf{(A3)} if we take the trimming threshold to be $\tau = d_{\floor{tn}}$ for some $t \in (0,1)$, e.g. the median singular value. Compared to the PCA adjustment, it has an advantage that one does not need to estimate the number of confounding variables from the data. Moreover, it does not shrink first several singular values to $0$, but only to the necessary level. This more gradual shrinkage might lead to better performance especially if the signal $X\beta$ is more aligned with the first few singular vectors.

\paragraph{Lava} The mapping $d_i \to \sqrt{n\lambda_2}d_i/\sqrt{n\lambda_2 + d_i^2}$ used in the Lava algorithm \citep{lava2017} satisfies the conditions \textbf{(A2)} and \textbf{(A3)} as well, since the transformed singular values $\tilde{d}_i$ are quite close to the ones for the Trim transform $\tilde{d}_i = \min(d_i, \tau)$, for an appropriate choice of $\tau$: 
$$\frac{1}{2}\min(d_i, \sqrt{n\lambda_2}) \leq \frac{\sqrt{n\lambda_2}d_i}{\sqrt{n\lambda_2 + d_i^2}} \leq \min(d_i, \sqrt{n\lambda_2}).$$ 
This also reveals how to choose the penalty $\lambda_2$ in Lava: $\lambda_2 = \tfrac{1}{n}d_{\floor{\min(n, p)/2}}^2$ and $\lambda_1$ can be chosen by cross-validation. This transformation has the property that it is smoother than the Trim transform. We note that with this comment and Corollary \ref{optimal_rate_perturbed}, we have established the standard Lasso $\ell_1$-error rate for Lava for estimating the sparse parameter $\beta$ in a high-dimensional regression model; such result is not given in \citet{lava2017}.

\paragraph{Puffer transformation} For the Puffer transform  \citep{jia2015preconditioning}, where we map all singular values to a constant $d_n$ (because of homogeneity it does not matter to which constant we map it, but we have assumed w.l.o.g. that $\tilde{d}_i \leq d_i$, so we need to map them to $d_n$), the assumption \textbf{(A2)} is easily satisfied. However, for \textbf{(A3)} we need to have $d_n^2 = \Omega_p\left(\lambda_{\min}(\Sigma)\,p\right)$. From \citet{vershynin2010introduction}, we have that this holds only if $\liminf \frac{p}{n} > 1$, i.e. the Puffer transform will not work well if $n$ and $p$ are close.

\paragraph{Step function} The justification of the assumptions \textbf{(A2)} and \textbf{(A3)} for Trim transform apply as well for the step function $\tilde{d}_i = \tau \1(d_i > \tau)$ with the same threshold $\tau$. However, unnecessarily shrinking singular values might cause worse performance than for the Trim transform.

\subsection{Low dimensional case: $n > p$} \label{lowdim}
The statement of Theorem \ref{general_bound} still holds in the low-dimensional case $n > p$. However, $\tfrac{1}{n}\norm{\tilde{X}b}_2^2$ will now be of larger order than $\lambda$. We have that $\lambda_{\max}(\tilde{X}) = \O_p(\sqrt{n})$, compared to $\sqrt{p}$ before (see Lemma \ref{averagesingularvalue}), which under the assumption \textbf{(A1')} gives us that $\tfrac{1}{n}\norm{\tilde{X}b}_2^2 = \O(\norm{b}_2^2) = \O(\tfrac{s\sigma^2 \log p}{p})$. Therefore, the second term in the bound of Theorem \ref{general_bound} will be too large in comparison with the first term.

Fortunately, from the remark below Theorem \ref{general_bound}, we see that by taking larger $\lambda$, we can decrease the rate of the second term. If the perturbation term $\tfrac{1}{n}\norm{\tilde{X}b}_2$ gets larger than the standard penalty rate, as it is the case when $n > p$, it is better to penalize more. One gets in this case:
$$\norm{\hat{\beta} - \beta}_1 = \O_p\left(\frac{s\sigma}{\lambda_{\min}(\Sigma)}\sqrt{\frac{\log p}{n}} + \frac{\sqrt{s}\norm{b}_2}{\sqrt{\lambda_{\min}(\Sigma)}}\right)$$
which by Lemma \ref{perturbationsize} in the confounding model, under the dense confounding assumption \textbf{(A1)}, becomes:
$$\norm{\hat{\beta} - \beta}_1 = \O_P\left(\frac{s\sigma}{\lambda_{\min}(\Sigma)}\sqrt{\frac{\log p}{n}} + \frac{\sqrt{s}\sigma}{\sqrt{\lambda_{\min}(\Sigma)}\sqrt{p}}\right).$$

One can not expect the same error rate as in the high-dimensional setting, since this would imply that, for fixed $p$, the error converges to $0$ as $n \to \infty$ which can not happen because the error is not only due to the randomness of the sample data, but also due to the coefficient perturbation $b$. The perturbation $b$ only depends on how the confounding variables affect the predictors and not on the number of data points and thus one can not expect consistency for a fixed $p$. However, we see that the estimator is consistent when $n,p \to \infty$. The more predictors we have, the more is the effect of the confounding variables spread out.

This is also illustrated in Figure \ref{ngrows}, where we can see that even though the error decreases as we increase the number of data points, it still seems to have a nonzero limit. However, the error is small, especially in comparison with the standard Lasso, and there is a benefit in using our method.

\begin{figure}[h]
\centering
\hspace*{-0.5cm}
\includegraphics[scale=0.75]{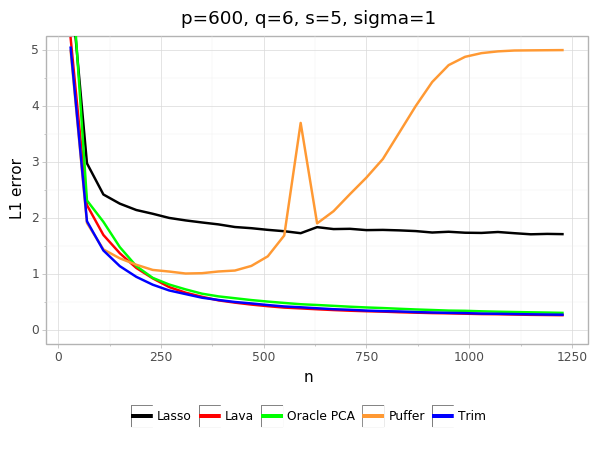}
\caption{Dependence of the estimation error $\|\hat{\beta} - \beta\|_1$ on the sample size $n$ for different spectral transformations and data generated from the confounding model, including the case $p < n$, as described in Section \ref{simsetting}.}
\label{ngrows}
\end{figure}

\section{Empirical Results}
We present here some empirical results for simulated and real data.

\subsection{Simulations}
We demonstrate the performance of various spectral transformations for estimating the coefficient vector $\beta$ with a subsequent use of the Lasso: Trim transform, Lava, Puffer and PCA adjustment. We investigate the cases when the perturbation $b$ arises from hidden confounding and when it is randomly sampled. 

\subsubsection{Setting} \label{simsetting}
We generate the data from the confounding model (\ref{confounding_model}). We take $\Sigma_E = \sigma_E^2I_p$, where $\sigma_E = 2$ and $\beta = (1, 1, 1, 1, 1,0,\ldots,0)$, so $s=5$. For a fixed number $q$ of hidden confounders, we sample the coefficients $\Gamma_{ij}$ and $\delta_i$ independently as standard normal random variables. By default, we take $q = 6$. Unless stated otherwise, we use the noise level $\sigma=1$ as the standard deviation of $\epsilon$. Finally, the sample size is set to be $n=200$ and the dimensionality of the predictors is $p=600$ as the default value. All results are based on $N=2^{12}=4096$ independent simulations.

It is also interesting to consider the perturbed linear model (\ref{perturbed_model}). We do not generate data from this model directly, but we will modify the underlying perturbation term $b$ which is implicit in the confounding model by formula (\ref{perturbation_formula}). This way we can compare the results obtained for the confounding model and the perturbed linear model directly with each other. We replace $b$ by $Qb$ where $Q$ is a random rotation matrix so that the new perturbation has the same size, but with uniformly random direction. We note that the resulting distribution is the same as of the perturbed linear model (\ref{perturbed_model}), where rows of $X$ are drawn from $N(0, \Sigma)$, where $\Sigma = \Gamma^T\Gamma + I_p$, and $b$ is drawn uniformly from a ball of radius $\norm{(\Gamma^T\Gamma + I_p)^{-1}\Gamma^T \delta}_2$.

\subsubsection{Choosing \texorpdfstring{$\lambda$}{lambda}} \label{choosing_lambda}
In practice we encounter the problem of choosing the penalty level $\lambda$ for the Lasso after applying a spectral transformation. The results of Theorem \ref{optimalrate} and Corollary \ref{optimal_rate_perturbed} give us that one can use the standard theoretical penalty rate $\lambda \asymp \sigma \sqrt{\tfrac{\log p}{n}}$ to get the desired error rate of our estimator. In practice one often resorts to using cross-validation (CV) for choosing the penalty parameter rather than using the theoretical value, especially since $\sigma$ is unknown. 

\begin{figure}[h]
\centering
\hspace*{-0.5cm}
\includegraphics[width=1.02\linewidth]{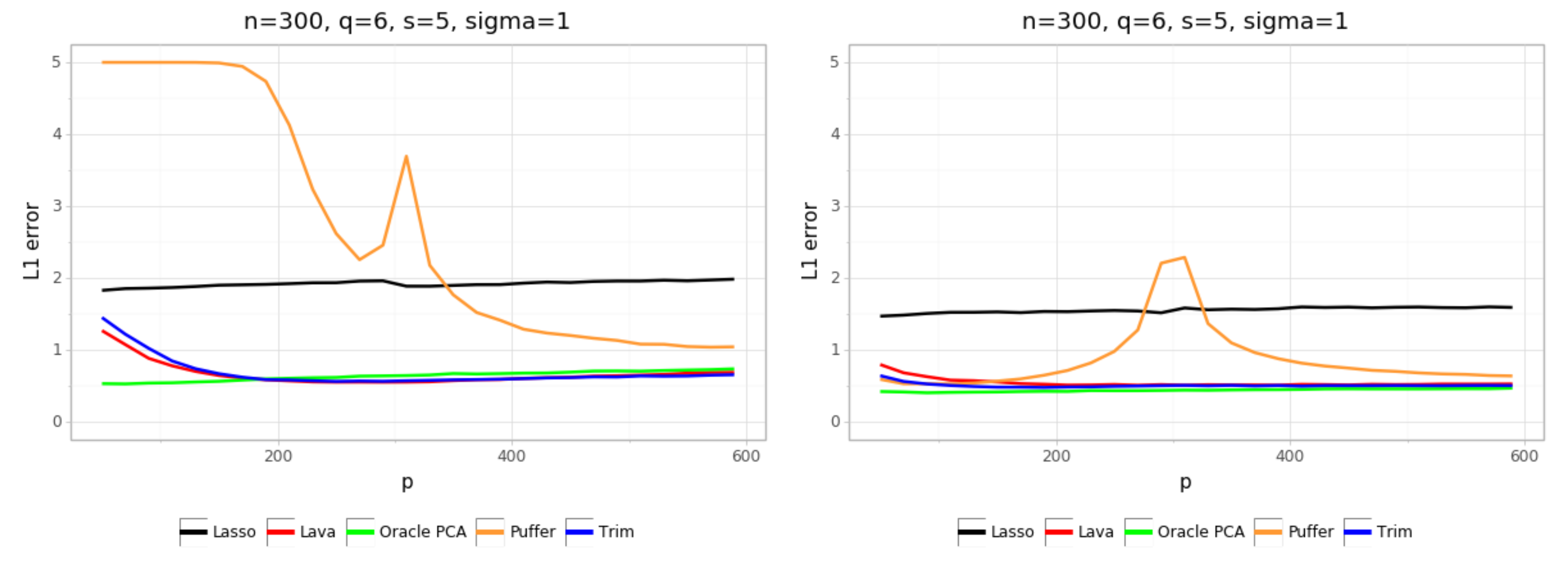}
\caption{Dependence of the estimation error $\|\hat{\beta} - \beta\|_1$ on the number of predictors $p$ for different spectral transformations and data generated from the confounding model \eqref{confounding_model}, as described in Section \ref{simsetting}. In the left plot, the penalty is chosen by cross-validation, whereas in the right plot we use the oracle value for which the estimation error is minimal.}
\label{pgrows}
\end{figure}

However, one needs to be careful in presence of confounding variables; in this case the coefficient vector $\beta + b$ describes the data better than $\beta$, which we are trying to recover. Therefore, cross-validation tends to choose a smaller value of $\lambda$ than the optimal for recovering $\beta$. This is illustrated in the Figure \ref{pgrows}, where we see that, for example, the Puffer transform is significantly affected by this choice of $\lambda$. For recovering $\beta$ in practice, it might be better to increase slightly the value of $\lambda$ chosen by cross-validation \citep{janzing2018detecting}. But on the other hand, smaller $\lambda$ gives us a larger set of variables, which might be beneficial for variable screening.

In all simulations, unless stated otherwise, the penalty level is chosen by cross-validation. This choice does not seem to worsen the performance of the Trim transform or Lava a lot, as one can see in Figure \ref{pgrows} and Figure \ref{misspecification}, and it is of great practical importance since the oracle value of $\lambda$, i.e.\ the one for which $\norm{\hat{\beta}_\lambda - \beta}_1$ is smallest, can not be directly determined from the data.

\subsubsection{Results}
Here we present the results of the simulations for both the confounding model and the perturbed linear model. A fundamental difference between them is that the coefficient perturbation arising from the confounding model is pointing towards the singular vectors of $X$ corresponding to the large singular values (see Figure \ref{projection}). This makes $\norm{Xb}_2$ larger for a fixed $\norm{b}_2$, and in this case the estimation error will be larger. On the other hand, in this case we can improve our accuracy more compared to the plain Lasso by shrinking large singular values, as will be shown below.

\begin{figure}[h]
\centering
\hspace*{-0.3cm}
\includegraphics[width=1.02\linewidth]{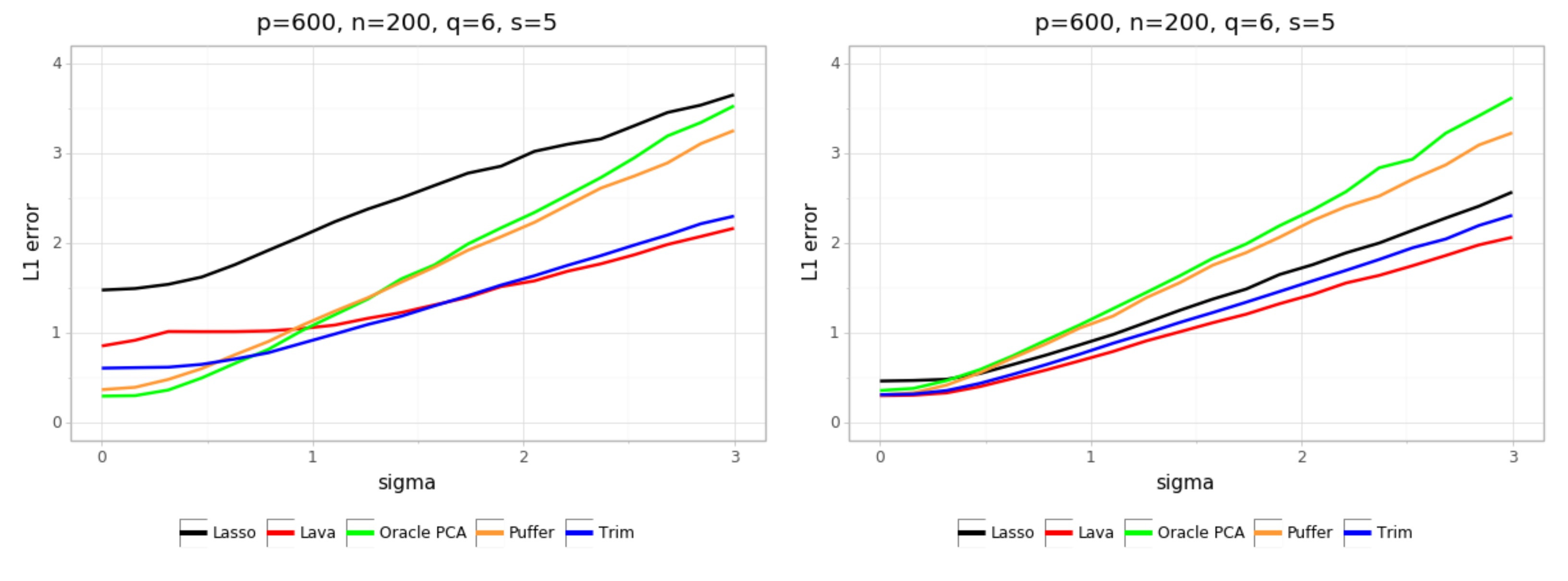}
\caption{Dependence of the estimation error $\|\hat{\beta} - \beta\|_1$ on the size of the noise for different spectral transformations for confounding model (left) and the perturbed linear model (right), as described in Section \ref{simsetting}.}
\label{sigma}
\end{figure}

\paragraph{Noise versus perturbation}
In the left plot in Figure \ref{sigma} we can see how the estimation error changes depending on the size of the noise $\sigma$ in the confounding model. When $\sigma$ is small, the perturbation $b$ has the biggest effect on the error. On the other hand, if $\sigma$ is large, then the influence of the perturbation $b$ becomes less pronounced.

We can see that the standard Lasso is affected a lot by the coefficient perturbation, whereas the Puffer transform and the PCA adjustment are affected more by the additive noise than the Lava and the Trim transform, since the slopes of the corresponding curves are steeper. The higher variance of the Puffer transform is most evident in Figure \ref{ngrows} and Figure \ref{pgrows}; when $n, p$ are close to each other, some of the singular values of $X$ become quite small and thus mapping them to a constant can inflate the error $\epsilon$ in the corresponding directions by a lot. We can observe that the oracle PCA adjustment, which removes exactly the $q$ largest singular values of $X$, works well, especially when $\sigma$ is small. For larger $\sigma$, we see that Trim transform and Lava work slightly better since they do not remove that much of the signal.

In the right plot of Figure \ref{sigma}, we have randomized the direction of $b$ while keeping everything else constant, as described in Section \ref{simsetting}. This then corresponds to a model with random perturbation $b$, but no specific further structure in terms of confounding. We can see a substantial improvement of the standard Lasso: in hindsight this shows that the Lasso is very sensitive to confounding variables but much less so to perturbation of sparsity. Also, it is worth noting that the PCA adjustment method is now consistently worse than the Trim transform or Lava, since the projection of $b$ onto the span of the first $q$ singular vectors is not that large anymore.

\begin{figure}[h]
\centering
\hspace*{-0.3cm}
\vspace*{-0.3cm}
\includegraphics[scale=0.6]{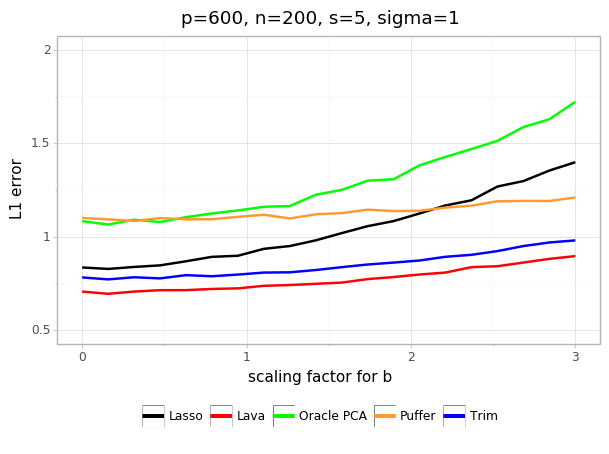}
\caption{Dependence of the estimation error on the size of the perturbation vector $b$ for different spectral transformation for the perturbed linear model, as described in Section \ref{simsetting}.}
\label{inflate_b}
\end{figure}

We can see more clearly the bias-variance tradeoff for different spectral transformations in Figure \ref{inflate_b}, where we have taken the rotated coefficient perturbation $b$, as in the right plot of Figure $\ref{sigma}$ and then artificially scaled it by a chosen constant. For a very small $b$, we see that Puffer and PCA adjustment have somewhat worse performance. As $b$ increases, Trim transform and Lava reduce the bias caused by $b$ much better than the Lasso. We can also see that the PCA adjustment does not reduce the bias as much, but its performance would be significantly better if $b$ was not rotated, but aligned with the top several principal components as in the confounding model, see Figure \ref{sigma}.

\paragraph{Number of confounding variables}
In Figure \ref{numconfounders} we can see how the estimation error depends on the number $q$ of confounding variables. As above, we see that the Lasso is severely affected by the presence of confounding variables. The Puffer transform performs reasonably well since $n$ and $p$ are different enough and the Trim transform and Lava exhibit similar and good performance in all cases.

PCA adjustment works well for the confounding model if we correctly guess the number of confounding variables. In the left plot in Figure \ref{numconfounders} we can clearly see how the estimation error is affected by the misspecification of the number of the principal components we remove. The oracle PCA method, which removes exactly $q$ principal components, performs reasonably well, particularly for smaller values of $q$. However, if we overestimate or especially if we underestimate the number of confounding variables, the estimation error will become significantly worse compared to the Trim transform or Lava.

\begin{figure}[h]
\centering
\hspace*{-0.3cm}
\includegraphics[width=1.02\linewidth]{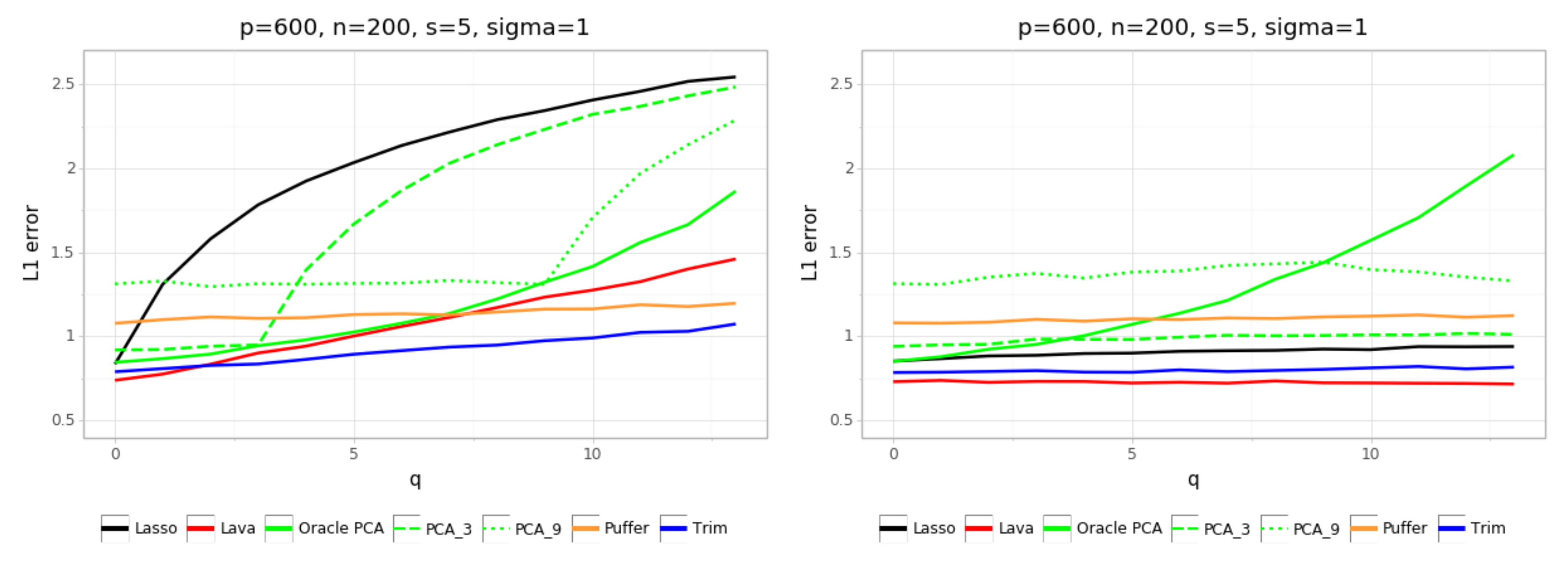}
\caption{Dependence of the estimation error $\|\hat{\beta} - \beta\|_1$ on the number of confounding variables for different spectral transformation for confounding model (left) and the perturbed linear model (right) as described in Section \ref{simsetting}.}
\label{numconfounders}
\end{figure}

\paragraph{Method robustness}
We are interested in whether there are any disadvantages in using the spectral transformations if we wrongly think that there is some hidden confounding or that the sparse coefficient has been perturbed.

In Figure \ref{misspecification} we display the estimation error for the confounding model as in Figure \ref{numconfounders}, but where the coefficient bias $b$ has been set to $0$, i.e.\ this is a standard sparse linear model with $X$ being generated from the spiked covariance model.

There is no indication for relevant differences between the performances of the Trim transform, Lava and the Lasso. The Lasso performs slightly better for larger values of $q$ and slightly worse for smaller $q$. It is worth noting that on this plot the estimation error starts to decrease as $q$ increases, which is due to a scaling issue. This happens because the variance of $X$ increases as $q$ increases, since $\Sigma = \Gamma^T \Gamma + \Sigma_E$, thus effectively increasing the signal to noise ratio. PCA adjustment seems to be affected most by the choice of $\lambda$, especially for larger $q$ since its shrinkage is larger in this case, see Figure \ref{misspecification}. With the oracle choice of the penalty level, its performance is very similar to the performance of the Lasso.

Our empirical results support theoretical evidence, which showed that it is safe to use wisely chosen spectral transformations such as the Trim transform or the Lava. If there are any confounding variables present, there is a large improvement over the standard Lasso. On the other hand, if there are no confounding variables, the Trim transform or Lava will have about the same performance as the Lasso. Therefore, our method can be thought of as an easy to use modification of the Lasso which is robust to hidden confounding.

\begin{figure}[h]
\centering
\hspace*{-0.5cm}
\includegraphics[width=1.02\linewidth]{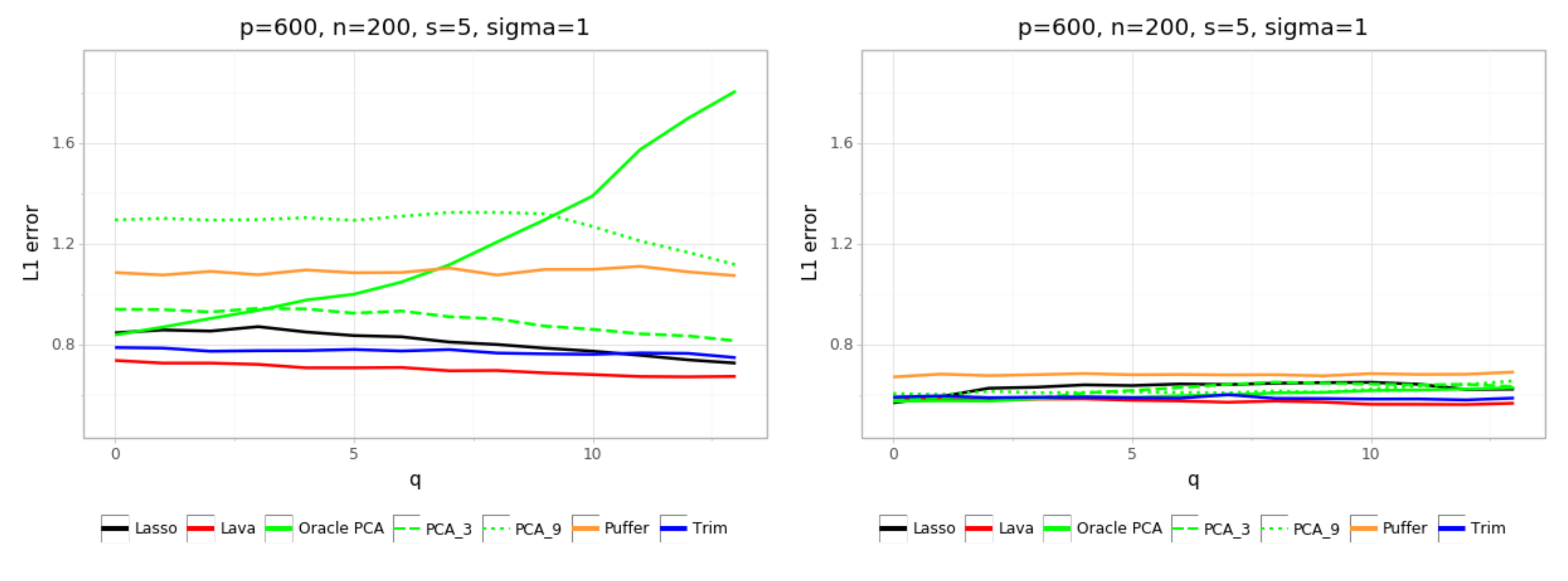}
\caption{Size of the estimation error $\|\hat{\beta} - \beta\|_1$ for a sparse linear model where $\Sigma = \Gamma^T\Gamma + I_p$, i.e. the confounding model with the induced perturbation $b$ set to $b = 0$. The penalty level $\lambda$ is either chosen by cross-validation (left) or taken to be the oracle value, which minimizes the $\ell_1$-error (right).}
\label{misspecification}
\end{figure}

\subsection{Application to genomic dataset}
In this section we demonstrate the robustness of our method against hidden confounders on a real genomic dataset where we have certain knowledge about the confounding variables. We inspect various spectral transformations in combination with the Lasso and evaluate the differences between the estimates for the original data set and the one where the confounding variables have been adjusted for.

\subsubsection{Gene expression dataset}
We have obtained data from the GTEx Portal (\url{http://gtexportal.org}). The GTEx project provides large-scale data with an aim to help the scientific community to study gene expression, gene regulation and their relationship to genetic variation. It provides gene expression data from 11,688 samples collected postmortem from 53 different tissues of 714 human donors. 

Gene expression is a process in the cell in which the information stored in a certain gene is used for the synthesis of gene products such as proteins. In the GTEx Project it was quantified by the amount of the mRNA in the cell which was created from this gene. Gene expression differs among different people and among different cells within the human body. The type of the cells is determined by the gene expression within them; even though the DNA in all cell nuclei is the same, cells in different tissues behave and look differently and perform significantly different tasks. Gene expression is also affected by the genetic variation and determining the expression quantitative trait loci (eQTL), which are parts of genome which explain the variation in the gene expression, is a very important problem which will help to understand the relationship between genetic variation and different phenotypes.

\subsubsection{Setting}
We use the fully processed, filtered and normalized gene expression matrix for the skeletal muscle tissue. We consider the gene expression of $p = 14'713$ protein-coding genes measured from $n = 491$ samples. For our purpose, an important aspect of this dataset is that there are also $q = 65$ different covariates provided, which are proxies for the hidden confounding variables. They include genotyping principal components and PEER factors. We can thus obtain the deconfounded data by regressing out these given covariates.

The left panel of Figure \ref{expression-svd} displays the singular values of the initial data matrix. We see that the first several singular values are substantially larger than the rest which suggests a possible existence of hidden confounders. In the right part of Figure \ref{expression-svd} we can see the singular values of the deconfounded data matrix where we have regressed out all of the $q = 65$ covariates which are provided as confounding proxies. 

\begin{figure}[h]
\centering
\hspace*{-0.4cm}
\includegraphics[scale=0.13]{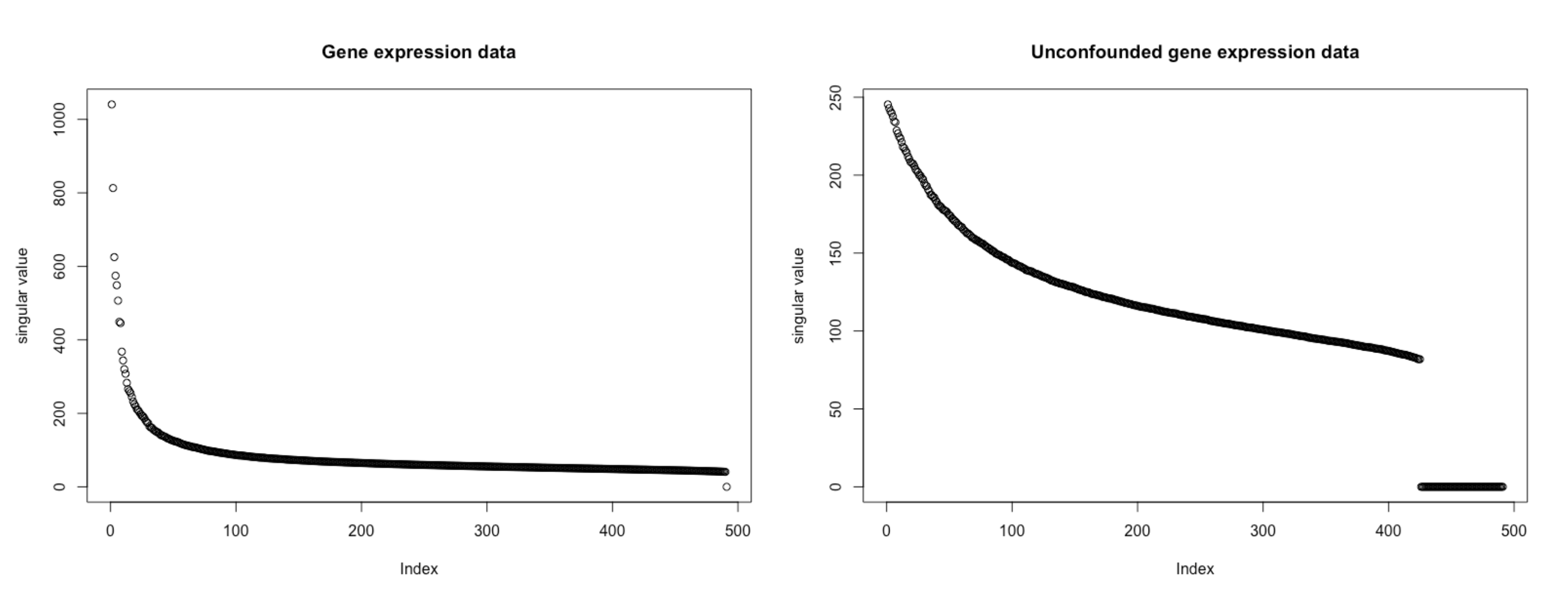}
\caption{Singular values of the gene expression data matrix for skeletal muscle tissue before (left) and after (right) regressing out the provided $q = 65$ confounding covariates.}
\label{expression-svd}
\end{figure}

We are going to explore now the robustness of the Lasso, Trim transform, and Lava against hidden confounders by comparing the estimates based on the original and the deconfounded data. 
For a fixed value of $k$, we regress out first $k$ given confounder proxies from the original gene expression data matrix $X$ in order to get the matrix $X^{(k)}$ and we randomly choose one column to represent the response $Y$. We are thus trying to explain the expression of one gene by the expressions of other genes.

For every $s = 1, \ldots, 20$, we apply the given method on $X$ and $X^{(k)}$ with the regularization $\lambda$ chosen as the largest value such that the support size of $\hat{\beta}$ equals a prespecified value $s$. This leads to estimates $\hat{\beta}_s$ and $\hat{\beta}_s^{(k)}$. We measure the dissimilarity of the corresponding supports by $J(\text{supp}\, \hat{\beta}_s,\, \text{supp}\, \hat{\beta}_s^{(k)})$, where $J$ is the Jaccard distance: 
$$J(A, B) = \frac{A \triangle B}{A \cup B}.$$

\begin{figure}[h]
\centering
\hspace*{-0.4cm}
\includegraphics[scale=0.118]{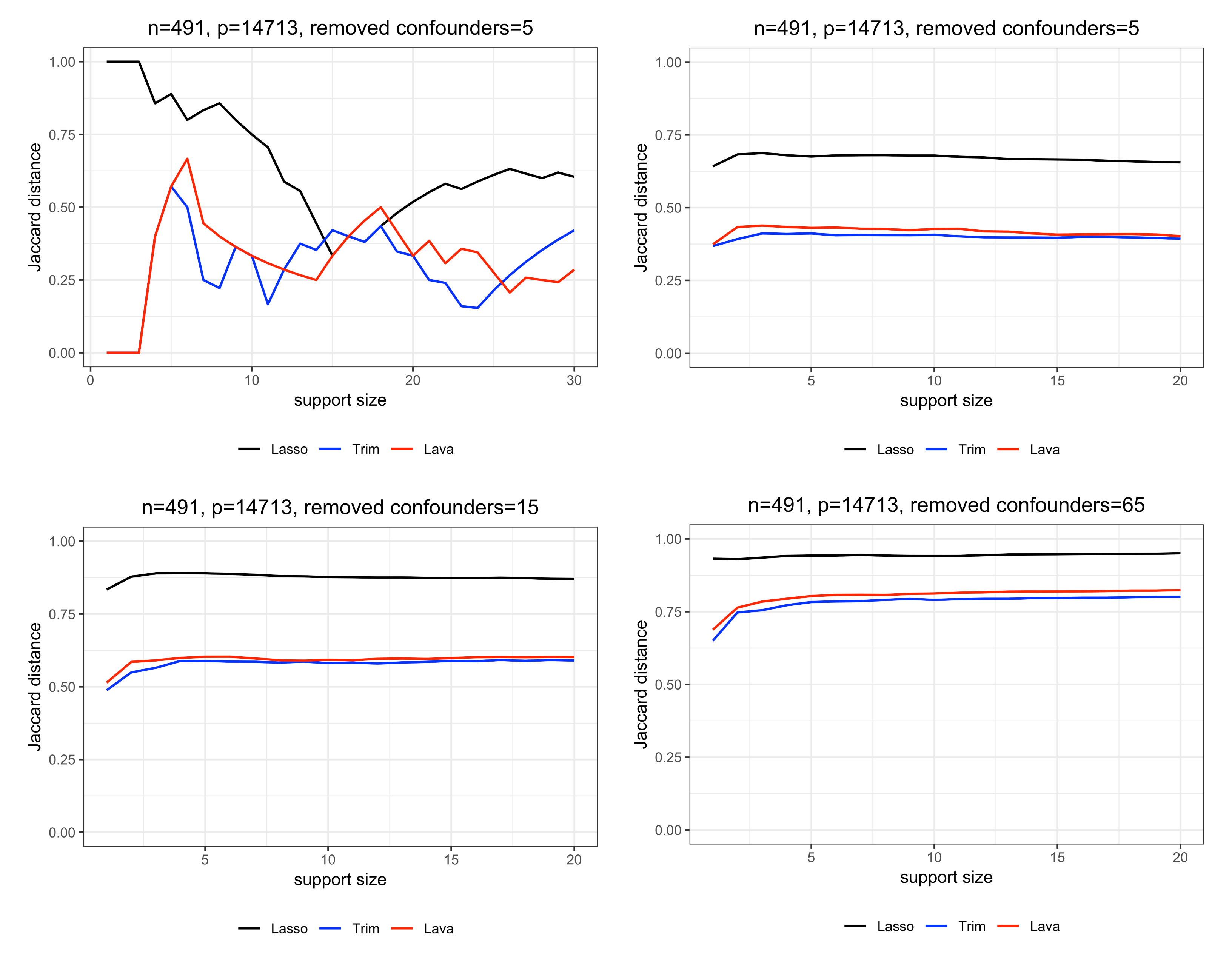}
\caption{Jaccard distance of the supports of the estimates based on the original and deconfounded data for one randomly chosen response (top left). Jaccard distance, averaged over $500$ randomly chosen responses, of the supports of estimates based on the original data and data with $5$ (top right), $15$ (bottom left) and $65$ (bottom right) confounder proxies removed.}
\label{jaccard}
\end{figure}

\subsubsection{Results}
In the top left image in Figure \ref{jaccard}, we can see the difference of the estimates for the original and the deconfounded data, where $5$ randomly chosen confounding variables have been removed and the response $Y$ is the expression of a randomly chosen gene. We can see that the Jaccard distance for the Lasso is closer to $1$, indicating that the estimated support sets are very different and almost disjoint; The Trim transform and Lava are much more robust to the hidden confounders and we see that the Jaccard distance between the estimates based on confounded and deconfounded data is much smaller.

In order to make sure that the choice of response $Y$ did not affect the results, we have repeated this experiment for $500$ randomly chosen genes and averaged the obtained results. The results are also displayed in Figure \ref{jaccard}. We can see that, as we increase the number $k$ of confounding variables which we regress out, the Jaccard distance for all methods is increasing. This is to be expected since $X^{(k)}$ and $X$ are becoming more different as we increase $k$. However, we can infer that the Trim transform and Lava are consistently better than the Lasso, exhibiting also in this real dataset the robustness against confounding variables.

\section{Discussion}
We propose to add robustness against hidden confounding variables by employing a wisely chosen spectral transformation before using the Lasso or other high-dimensional sparse regression techniques. There is essentially nothing to lose but much to be gained which is in line with the typical argument of robustness \citep{huber2011robust} We can also take directly the viewpoint of deconfounding before performing further analysis: this is the more common thinking in many applications where hidden confounding is expected to happen, a prime example being genetics \citep{novembre2008interpreting}.

The confounding issue in the context of linear models can be represented and analyzed as a regression problem with coefficient $\beta + b$; the coefficient $\beta$ is the true underlying parameter in absence of confounding variables, while the perturbation $b$ is due to the confounding. We develop theory for a linear model with regression parameters $\beta + b$ where $\beta$ is sparse and the perturbation $b$ sufficiently small, a condition satisfied when the confounding is sufficiently 'dense' in the sense that each confounding variable affects many predictors. We show that certain spectral transformations, such as the Trim transform or the PCA adjustment, in conjunction with using the Lasso afterwards, achieve the same $\ell_1$-convergence rate of the $\|\hat{\beta} - \beta\|_1$ as the Lasso for the linear model without confounding; see Section \ref{theory} and Theorem \ref{optimalrate}. Such a theoretical result is entirely new and covers also the Lava method \citep{lava2017}. As a consequence, the theoretical result also establishes spectral deconfounding as an excellent method for removing the effect of dense hidden confounders in high-dimensional settings.

Another advantage of our approach is its simplicity: it consists of just one simple pre-transformation step before using the Lasso. It requires the computation of the SVD of the design matrix which has computational complexity of $\O(\min(n^2p,\, np^2))$ and can be done in a few lines of code.

The topic of deconfounding has not received too much attention, despite its practical importance \citep{greenland1999confounding, brookhart2010confounding}. Here we have shown that it is possible and easy to protect against hidden dense confounding in the case of linear regression. Similar ideas might be powerful as well for more complicated models.

\section*{Acknowledgements}
We would like to thank Gian Thanei and Benjamin Frot for providing processed data for our biological applications.

\appendix
\section{}

\subsection{Compatibility condition after transformation}
We discuss here in more detail the assumption \textbf{(A3)}, which is the compatibility condition \citep{buhlmann2011statistics} for the transformed design matrix $\tilde{X}$.

This assumption means that the compatibility constant $\phi_{\hat{\Sigma}}$ does not substantially decrease after applying our transformation $F$. Intuitively, we want to show that by transforming the singular values, we have not shrunk our signal $X\beta$ too much. For Trim transform and PCA adjustment, this means that the active set $X_S$ is not aligned too much with the direction of the singular vectors corresponding to the large singular values, since those are the directions along which we are substantially shrinking. 

In order to proceed, we need to understand how exactly $\phi_{\tilde{\Sigma}}$ depends on the transformed singular values in $\tilde{D}$. The answer to this question depends delicately on the singular vectors $V$ of $X$ as well. The following bound helps us to understand the behaviour of the compatibility constant depending on the transformed singular values.
\begin{restatable}{lemma}{compatibilityineq} \label{compatibility-ineq}
Consider an arbitrary spectral transformation $F$ as in (\ref{spectraltransformation}). Let $1 \leq k < r=\min(n,p)$ be an arbitrary integer. Then:
$$\phi_{\tilde{\Sigma}}^2 \geq \sum_{i=1}^r \frac{1}{n}\tilde{d}_{(i)}^2(\phi^2_{M_i} - \phi^2_{M_{i-1}}) \geq \frac{1}{n}\tilde{d}_{(k)}^2\phi^2_{M_k},$$
where $M_k = \left[V_{(1)}, \ldots, V_{(k)}\right]\left[V_{(1)}, \ldots, V_{(k)}\right]^T$.
\end{restatable}

We can use Lemma \ref{compatibility-ineq} for the Trim transform with $k = \floor{tn}$ or the PCA adjustment where we map $\hat{q} = o(k)$ singular values to $0$. Then it suffices to show that $\tilde{d}_{(k)}^2$ and $\phi^2_{M_k}$ are sufficiently large in order to show \textbf{(A3)}. This means that certain proportion of the singular values is still large (of order $\sqrt{p}$) after shrinking, and that the direction of the corresponding singular vectors is not too unfavourable.

The Lemma \ref{singularvalues} from the main part of the paper shows that for quite a wide range of settings we have that $d_{\floor{tn}}^2 = \Omega_p(\lambda_{\min}(\Sigma)p)$, which means that after applying the Trim transform with $\tau = d_{\floor{tn}}$, or the PCA adjustment which shrinks $o(n)$ singular values to zero, we still have $\tilde{d}_{(\floor{tn})}^2 = \Omega_p(\lambda_{\min}(\Sigma) p)$. It therefore suffices to show that $\phi_{M_{\floor{tn}}}^2 = \Omega_p\left(\frac{n}{p}\right)$ in order to show \textbf{(A3)} from Lemma \ref{compatibility-ineq}. This is not simple to show, as distribution of $V$ is quite complicated for general $\Sigma = \Cov(X)$, but this always holds under the uniformity condition described in the next lemma.

\begin{restatable}{lemma}{uniformV} \label{M-compatibility}
If $p > n$ and $V$ has a uniform distribution on the Stiefel manifold, then for any $k = \Omega(n)$, we have
$$\phi_{M_k}^2 = \Omega_p\left(\frac{n}{p}\right).$$
\end{restatable}

This uniformity assumption is sensible to make since it will be true under any of the two following scenarios: the first is that the components of $X$ are i.i.d.\ normal random variables; the second is that the singular vectors of $\Sigma$ have the uniform distribution on the space of orthogonal matrices themselves. This, for example, might happen in the confounding model (\ref{confounding_model}), when $\Cov(E) = \sigma_E^2I_p$ and the rows of $\Gamma$ have rotationally invariant distribution, i.e. $\Gamma Q$ has the same distribution as $\Gamma$ for any orthogonal matrix $Q \in \R^{p \times p}$. 

The uniformity assumption is sufficient, but not necessary for the assumption (\textbf{A3}) to hold. In the main part it was shown to hold in a certain asymptotic regime for the Trim transform and PCA adjustment. We believe however that the compatibility condition holdS in a very broad range of asymptotic regimes.

\subsection{Proofs}

Here we provide the proofs of all results stated in this paper.

\generalbound*
\begin{proof} Denote by $\beta^0$ the true coefficient vector. 

Since $\hat{\beta}$ minimizes $\frac{1}{n}\|\tilde{Y} - \tilde{X}\beta\|^2_2 + \lambda\|\beta\|_1$, we have:
\begin{align*}
\frac{1}{n} \|\tilde{Y} - \tilde{X}\hat{\beta}\|_2^2 + \lambda\|\hat{\beta}\|_1 &\leq \frac{1}{n} \|\tilde{Y} - \tilde{X}\beta^0\|_2^2 + \lambda\|\beta^0\|_1 \\ 
\frac{1}{n} \|\tilde{X}(\hat{\beta} - \beta^0)\|_2^2 + \lambda\|\hat{\beta}\|_1 &\leq \frac{2}{n} (\tilde{Y} - \tilde{X}\beta^0)^T\tilde{X}(\hat{\beta}-\beta^0) + \lambda\|\beta^0\|_1 \\
&\leq \frac{2}{n} \tilde{\epsilon}^T\tilde{X}(\hat{\beta}-\beta^0) + \frac{2}{n} b^T\tilde{X}^T\tilde{X}(\hat{\beta}-\beta^0) + \lambda\|\beta^0\|_1 \\
\frac{1}{n} \|\tilde{X}(\hat{\beta} - \beta^0 - b)\|_2^2 + \lambda\|\hat{\beta}\|_1 &\leq \frac{2}{n} \tilde{\epsilon}^T\tilde{X}(\hat{\beta}-\beta^0) + \frac{1}{n} \|\tilde{X}b\|_2^2 + \lambda\|\beta^0\|_1
\end{align*}

Let us work on the event $\{\|\frac{2}{n}\tilde{X}^T\tilde{\epsilon}\|_\infty \leq \tau\}$, which has probability at least $1 - 2p^{1-A^{2}/\left(32\max_{i} \Sigma_{ii}\right)}$$- p e^{-n/136}$ for $\tau = \lambda/2 = \frac{1}{2}A\sigma\sqrt{\frac{\log(p)}{n}}\lambda_{\max}(F)^2$, as it is shown in Lemma $\ref{proof:tau}$. On this event we have
$$\frac{2}{n}\tilde{\epsilon}^T\tilde{X}(\hat{\beta}-\beta^0) \leq \frac{2}{n}\|\tilde{X}^T\tilde{\epsilon}\|_\infty \|\hat{\beta}-\beta^0\|_1 \leq \tau \|\hat{\beta}-\beta^0\|_1$$
from H\"{o}lder's inequality. We now have:
$$\frac{1}{n} \|\tilde{X}(\hat{\beta} - \beta^0 - b)\|_2^2 + \lambda\|\hat{\beta}\|_1  \leq \tau\|\hat{\beta}-\beta^0\|_1 + \frac{1}{n}\|\tilde{X}b\|_2^2 + \lambda\|\beta^0\|_1$$
By using that $\beta^0_{S^c} = 0$, we get that
\begin{align*}
\frac{1}{n} \|\tilde{X}(\hat{\beta} - \beta^0 - b)\|_2^2& + (\lambda-\tau)\|\hat{\beta}_{S^c} - \beta^0_{S^c}\|_1 \notag \\
&\leq \tau\|\hat{\beta}_S-\beta^0_S\|_1 + \lambda\|\beta^0_S\|_1 - \lambda\|\hat{\beta}_S\|_1 + \frac{1}{n}\|\tilde{X}b\|_2^2\notag \\
&\leq (\lambda + \tau) \|\hat{\beta}_S-\beta^0_S\|_1 + \frac{1}{n}\|\tilde{X}b\|_2^2
\end{align*}

Let us now write
$$\phi_{\tilde{\Sigma}}(L,S) = \min_{\beta \in R(L,S)} \frac{\sqrt{\beta^T\tilde{\Sigma}\beta}}{\frac{1}{\sqrt{s}}\|\beta_S\|_1} > 0$$
where $R(L, S) = \{x: \|x_{S^c}\|_1 \leq L\|x_S\|_1\}$

We consider two cases:
\begin{itemize}
\item Case $1$: $\frac{1}{n}\|\tilde{X}b\|_2^2 \leq \lambda\|\hat{\beta}_S-\beta^0_S\|_1$
\item Case $2$: $\frac{1}{n}\|\tilde{X}b\|_2^2 \geq \lambda\|\hat{\beta}_S-\beta^0_S\|_1$
\end{itemize}

In the first case we have
$$\frac{1}{n} \|\tilde{X}(\hat{\beta} - \beta^0 - b)\|_2^2 + (\lambda-\tau)\|\hat{\beta}_{S^c} - \beta_{S^c}^0\|_1 \leq (2\lambda + \tau)\|\hat{\beta}_S-\beta^0_S\|_1$$
From this we see that the error $\hat{\beta} - \beta \in R(L, S) = \{x: \|x_{S^c}\|_1 \leq L\|x_S\|_1\}$ for $L = \frac{2\lambda + \tau}{\lambda - \tau}$, so we have:
\begin{align*}
\frac{1}{n} \|\tilde{X}(\hat{\beta} - \beta^0 - b)\|_2^2& + (\lambda - \tau)\|\hat{\beta} -\beta^0\|_1 \leq 3\lambda\|\hat{\beta}_S - \beta^0_S\|_1 \\
&\leq \frac{3\lambda\sqrt{s}\|\tilde{X}(\hat{\beta} - \beta^0)\|_2}{\sqrt{n}\phi_{\tilde{\Sigma}}(L, S)}\\
&\leq \frac{3\lambda\sqrt{s}\|\tilde{X}(\hat{\beta} - \beta^0 - b)\|_2}{\sqrt{n}\phi_{\tilde{\Sigma}}(L, S)} + \frac{3\lambda\sqrt{s}\|\tilde{X}b\|_2}{\sqrt{n}\phi_{\tilde{\Sigma}}(L, S)}\\
&\leq \frac{9\lambda^2s}{2\phi_{\tilde{\Sigma}}(L, S)^2} + \frac{1}{n} \|\tilde{X}(\hat{\beta} - \beta^0 - b)\|_2^2 + \frac{1}{n}\|\tilde{X}b\|_2^2
\end{align*}
by using the inequality $xy \leq \frac{x^2}{4} + y^2$ twice, which finally gives us
$$(\lambda - \tau)\|\hat{\beta} -\beta^0\|_1 \leq \frac{9\lambda^2s}{2\phi_{\tilde{\Sigma}}(L, S)^2} + \frac{1}{n}\|\tilde{X}b\|_2^2$$

In the second case we have
$$\frac{1}{n} \|\tilde{X}(\hat{\beta} - \beta^0 - b)\|_2^2 + (\lambda-\tau)\|\hat{\beta} -\beta^0\|_1 \leq \frac{3}{n}\|\tilde{X}b\|_2^2$$
So, regardless whether we are in the Case 1 or the Case 2, we get that
\begin{equation*}
(\lambda - \tau)\|\hat{\beta} -\beta^0\|_1 \leq \frac{9\lambda^2s}{2\phi_{\tilde{\Sigma}}(L, S)^2} + \frac{3}{n}\|\tilde{X}b\|_2^2
\end{equation*}

By dividing by $(\lambda - \tau) = \lambda/2$ we get the required inequality which is what we wanted.

It is interesting to note that we might get better rate of the second term in the case when $\norm{\tilde{X}b}_2$ has larger rate than it will be case in this paper, by taking $\lambda$ to be larger than $2\tau$. Since in this case, as we will now see, the penalty level depends on the unknown $b$, we decided to use the "standard" rate of $\lambda$.

By dividing by $(\lambda - \tau)$ and minimizing over $\lambda > \tau$, we get that the minimum value of the RHS of the bound is:
$$\frac{9s\tau}{\phi_{\tilde{\Sigma}}(L, S)^2} + \sqrt{\left(\frac{9s\tau}{\phi_{\tilde{\Sigma}}(L, S)^2}\right)^2 + \frac{54s\|\tilde{X}b\|_2^2}{\phi_{\tilde{\Sigma}}(L, S)^2n}}$$
which is achieved for 
$$\lambda = \tau + \sqrt{\tau^2 + \frac{2\phi_{\tilde{\Sigma}}(L, S)^2\|\tilde{X}b\|_2^2}{3sn}}$$

In the case when $b=0$ and $F=I_n$ (the usual Lasso regression), we indeed take $\lambda = 2\tau$. We can see that, when the coefficient perturbation is present, it is better to penalize more as this will remove the effect of the perturbation to some extent.

Since $L = \frac{2\lambda + \tau}{\lambda - \tau}$ and $\lambda \geq \lambda_{\min} \geq 2\tau$, we have $L \leq 5$ and then $$\phi_{\tilde{\Sigma}}(L, S) \geq \phi_{\tilde{\Sigma}}(5, S) = \phi_{\tilde{\Sigma}}$$

Finally, by using this and the inequality $\sqrt{x^2+y^2} \leq x + y$ where $x,y>0$, we get
$$\|\hat{\beta} - \beta^0\|_1 \leq \frac{18s\tau}{\phi_{\tilde{\Sigma}}^2} + \sqrt{\frac{54s\|\tilde{X}b\|_2^2}{\phi_{\tilde{\Sigma}}^2n}}.$$
\end{proof}

\begin{lemma} \label{proof:tau}
Let $A > 0$ be arbitrary constant. Let us define $$\tau = \frac{1}{2}A\sigma\sqrt{\frac{\log(p)}{n}}\lambda_{\max}(F)^2$$
Let $\epsilon \in \R^n$ be a vector consisting of i.i.d. sub-Gaussian random with mean zero and variance $\sigma^2$ independent of $X$. We have $$\P\left(\frac{2}{n} \|\tilde{X}^T\tilde{\epsilon}\|_\infty \leq \tau\right) \geq 1 - 2p^{1-A^2/(32\max_i \Sigma_{ii})} - pe^{-n/136}$$
\end{lemma}
\begin{proof}
Let us work on the event $\Omega_2 = \{\max_i \frac{\|X_i\|_2}{\sqrt{n}} \leq 2\max_i \Sigma_{ii}\}$, where $X_i$ is the i-th column of $X$. Since $X_{ji}$ is a mean zero sub-Gaussian random variable with variance $\Sigma_{ii}$, $X_{ji}^2$ satisfies Bernstein's condition with parameter $(8\Sigma_{ii}, 4\Sigma_{ii})$. Bernstein's inequality gives us that 
$$\P(\tfrac{1}{n}\|X_i\|^2 - \Sigma_{ii} > \Sigma_{ii}) \leq \exp(-\frac{n \Sigma_{ii}^2}{2(64\Sigma_{ii}^2 + 4\Sigma_{ii}^2)}) = e^{-\frac{n}{136}}.$$
From here it is easy to see that
$$\P(\Omega_2) \geq 1 - \sum_i \P(\tfrac{1}{n}\|X_i\|^2 > 2\Sigma_{ii}) = 1 - pe^{-\frac{n}{136}}$$

Conditionally on $X$, the components of $\zeta = \frac{2}{n}\tilde{X}^T\tilde{\epsilon} = \frac{2}{n}X^TF^2\epsilon$ are sub-Gaussian random variables since they are linear combinations of independent sub-Gaussian random variables; $\zeta_i$ is sub-Gaussian with mean zero and parameter $\sigma_i = \sigma\|\frac{2}{n}(X_i)^TF^2\|_2$, where $X_i$ is the i-th column of $X$.

From the tail bound for sub-Gaussian random variables, we now get: 
\begin{equation*}
\P(\|\zeta\|_\infty \leq \tau) \geq 1 - \sum_i \P(|\zeta_i| > \tau) \geq 1 - p\max_i 2\exp\left(-\frac{\tau^2}{2\sigma_i^2}\right) = 1 - 2\exp\left(-\frac{\tau^2}{2\max_i \sigma_i^2} + \log p\right)
\end{equation*}
We also have on event $\Omega_2$:
$$\max_i \sigma_i = \max_i \sigma\|\frac{2}{n}(X_i)TF^2\|_2 \leq \frac{2\sigma}{n}\lambda_{\max}(F)^2 \max_i\|X_i\|_2 \leq \frac{4\sigma}{\sqrt{n}}\lambda_{\max}(F)^2 \max_i\Sigma_{ii}$$
and plugging this and the expression for $\tau$ in the expression above gives us
$$\P(\|\zeta\|_\infty \leq \tau) \geq 1 - 2p^{1-\frac{A^2}{32\max_i \Sigma_{ii}}}.$$
Which gives us in general that $\P(\|\zeta\|_\infty \leq \tau) \geq 1 - 2p^{1-\frac{A^2}{32\max_i \Sigma_{ii}}} - pe^{-\frac{n}{136}}$, as required.
\end{proof}

\optimalrateperturbed*
\begin{proof}
From Theorem \ref{general_bound} we have 
$$\|\hat{\beta} - \beta\|_1 \leq C_1\frac{s\lambda}{\phi_{\tilde{\Sigma}}^2}  + C_2\frac{\|\tilde{X}b\|_2^2}{n\lambda}$$
From the assumptions \textbf{(A1)} and \textbf{(A2)} we get $$\norm{\tilde{X}b}_2^2 \leq \lambda_{\max}(\tilde{X})^2\norm{b}_2^2 \O(s\sigma^2\log p)$$
which gives the rate of the second term
$$\frac{\|\tilde{X}b\|_2^2}{n\lambda} = \O_p\left(s\sigma\sqrt{\frac{\log p}{n}}\right)$$
The assumption \textbf{(A3)} gives us that the rate of the first term is 
$$\frac{s\lambda}{\phi_{\tilde{\Sigma}}^2} = \O_p\left(\frac{s\sigma}{\lambda_{\min}(\Sigma)}\sqrt{\frac{\log p}{n}}\right)$$
which is what we wanted to show.
\end{proof}

\perturbationsize*
\begin{proof}
Let us write $C = \Gamma \Sigma_E^{-1/2}$ and let $C = U_CD_CV_C^T$ be its SVD. Then we can write:
\begin{align*}
\norm{b}_2^2 &= \norm{(\Gamma^T\Gamma + \Sigma_E)^{-1}\Gamma^T\delta}_2^2 = \norm{\Sigma_E^{-1/2}(C^TC+I_p)^{-1}C^T\delta}_2^2 \\
&= \norm{\Sigma_E^{-1/2}V_C(D_C^TD_C+I_p)^{-1}D_CV_C^T\delta}_2^2 \\
&\leq \lambda_{\max}(\Sigma_E^{-1}) \lambda_{\max}\left(V_C(D_C^TD_C+I_p)^{-1}D_CV_C^T\right)^2\norm{\delta}_2^2 \\
&\leq \lambda_{\min}(\Sigma_E)^{-1} \max_{i} \left( \frac{(D_C)_{ii}}{(D_C)_{ii}^2 + 1} \right)^2\norm{\delta}_2^2 \\
&\leq \lambda_{\min}(\Sigma_E)^{-1} \frac{1}{\lambda_{\min}(C)^2}\norm{\delta}_2^2 \\
&\leq \lambda_{\min}(\Sigma_E)^{-1} \frac{1}{\lambda_{\min}(\Sigma_E^{-1/2})^2\lambda_{\min}(\Gamma)^2}\norm{\delta}_2^2 \\
&\leq \frac{\lambda_{\max}(\Sigma_E)}{\lambda_{\min}(\Sigma_E)} \cdot \frac{\norm{\delta}_2^2}{\lambda_{\min}(\Gamma)^2} = \O\left(\frac{\norm{\delta}_2^2}{p}\right) = \O\left(\frac{\sigma^2}{p}\right).
\end{align*}

Now, if the rows or columns of $\Gamma$ are sub-Gaussian random variables with covariance matrix $\Omega$, we can write $\Gamma$ as $Z\Omega^{1/2}$ or $\Omega^{1/2}Z$ respectively, where the rows or columns of $Z$ are sub-Gaussian random variables with covariance $I$. In both cases we have from Theorems 5.39 and 5.58 from \citet{vershynin2016high} that there are constants $c$ and $C$ such that with probability $1 - 2\exp(-ct^2)$ we have
$$\sqrt{p} - C\sqrt{q} - t \leq \lambda_{\min}(Z) \leq \lambda_{\max}(Z) \leq \sqrt{p} + C\sqrt{q} + t$$
and thus since $\tfrac{p}{q} \to \infty$ we have
$$\lambda_{\min}(Z) = \Omega_p(\sqrt{p}),$$
which implies
$$\lambda_{\min}(\Gamma) = \Omega_p(\sqrt{p})$$
if $\lambda_{\min}(\Omega)$ is bounded from below, as we wanted.
\end{proof}

\averagesingularvalue*
\begin{proof}
We have 
$$\frac{1}{n}\sum_{i=1}^r d_i^2 = \Tr(\tfrac{1}{n}X^TX) = \Tr(\hat{\Sigma})$$
Since the rows of $X$ are sub-Gaussian random vectors, we get that $\norm{\tfrac{1}{n}X^TX - \Sigma}_\infty = \O_p\left(\sqrt{\frac{\log p}{n}}\right)$, as in Lemma \ref{proof:tau}. This gives us that 
$$|\Tr(\tfrac{1}{n}X^TX) - \Tr(\Sigma)| = \O_p\left(p\sqrt{\frac{\log p}{n}}\right).$$
Therefore, we have
$$\left|\frac{1}{n}\sum_{i=1}^r d_i^2 - \Tr(\Sigma)\right| = \O_p\left(p\sqrt{\frac{\log p}{n}}\right)$$ and the result follows since we have assumed that $\Tr(\Sigma) = \Omega(p)$.
\end{proof}

\begin{lemma} \label{proof:singularvalues}
Let $B \in \R^{p \times p}$ be a symmetric positive definite matrix and let $A \in \R^{n \times p}$ be arbitrary matrix, $n < p$. Let $\lambda_i(A)$ and $\lambda_i(AB)$ be the i-th largest singular values of $A$ and $AB$ respectively. Assume that the smallest singular value of $B$ is at least $1$. For $i \leq n$, we have $$\lambda_i(A) \leq \lambda_i(AB).$$
On the other hand, if we assume that the largest singular value of $B$ is at most $1$, we have for $i \leq n$ that $$\lambda_i(A) \geq \lambda_i(AB).$$
\end{lemma}
\begin{proof}
Let us first show the first statement. Let $e_1, \ldots, e_n$ i $f_1, \ldots, f_n$ be the left singular vectors of $A$ and $AB$ corresponding to the singular values in a decreasing order.
For $i=1$, since $\lambda_{\min}(B) \geq 1$, we have: $$\lambda_1(AB) \geq \|(AB)^Te_1\|_2 \geq \|BA^Te_1\|_2 \geq \|A^Te_1\|_2 = \lambda_1(A)$$ 

Let us proceed by induction. Since $\dim(U \cap V) \geq \dim(U) + \dim(V) - n$, we conclude that $F_k = \text{span}\left\{f_1,\ldots,f_k\right\}^{\perp}$ and $\text{span}\left\{e_1,\ldots e_{k+1} \right\}$ have a non-trivial intersection, so we can choose a unit vector $v = \sum_{j=1}^{k+1} \alpha_je_j \in F_k$. Since $\lambda_{k+1}(AB) = \max\{ (AB)^Tx : x \in F_k, \|x\|_2 = 1 \}$, we have:
$$\lambda_{k+1}(AB) \geq \|(AB)^Tv\|_2 \geq \|A^Tv\|_2 = \sqrt{\sum_{j=1}^{k+1} \alpha_j^2 \lambda_j(A)^2} \geq \lambda_{k+1}(A)$$ 
The second inequality holds because $\lambda_{\min}(B) \geq 1$ and the last because $\sum \alpha_j^2=1$ and $\lambda_i(A)$ are decreasing.

We can derive the second statement from the first one by considering $B \leftarrow B^{-1}, A \leftarrow AB$, since if $B$ has all singular values below $1$, $B^{-1}$ has all singular values at least equal to $1$.
\end{proof}

\PCAsingularvalue*
\begin{proof}
We can write $X = Z\Sigma^{1/2}$, where the rows of $Z$ are i.i.d. sub-Gaussian random variables and $\Cov(Z) = I_p$. From Theorem 5.39 from \citet{vershynin2016high} we now have that there exist constants $c$ and $C$ such that 
$$\sqrt{p} + C\sqrt{n} + t \geq \lambda_{\min}(Z)$$
with probability at least $1 - 2\exp(-ct^2)$. Therefore $\lambda_{\max}(Z) = \O_p(\sqrt{p})$.

From this, the assumption $\lambda_{\max}(\Sigma_E) = \O(1)$ and the second part of Lemma \ref{singularvalues}, we have that: 
$$\lambda_{q+1}(X) = \lambda_{q+1}(X) \leq \lambda_{q+1}(\Sigma^{1/2}) \lambda_{\max}(Z) \leq \lambda_{\max}(\Sigma_E)^{1/2} \lambda_{\max}(Z) = \O_p(\sqrt{p}),$$
as required.
\end{proof}

\singularvalues*
\begin{proof}
Let us write $X = Z \Sigma^{1/2}$ where $Z \in \R^{n \times p}$ is a matrix with i.i.d. sub-Gaussian rows with covariance matrix $I_p$. Let $\zeta_1 \leq \ldots \leq \zeta_n$ be the singular values of $Z$.

Since we can write $X = Z \Sigma^{1/2}$, by Lemma \ref{proof:singularvalues}, we have $d_k \geq \lambda_{\min}(\Sigma)^{1/2}\zeta_k$, so it suffices to show that $\zeta_k = \Omega_p(p^{1/2})$.

If we are under the first set of assumptions, we have from Theorem 5.39 from \citet{vershynin2010introduction} that there exist constants $c$ and $C$ such that 
$$\sqrt{p} - C\sqrt{n} - t \leq \zeta_n$$
with probability at least $1 - 2\exp(-ct^2)$. Therefore, since $\tfrac{p}{n} \to \infty$, we have $\zeta_k \geq \zeta_n = \Omega_p(p^{1/2})$, as required.

If we are under the second set of assumptions, we have that the entries of $Z$ are i.i.d. $N(0,1)$ variables. From Corollary 5.35 from \citet{vershynin2010introduction} that there exists constant $c$ such that 
$$\sqrt{p} - \sqrt{n} - t \leq \zeta_n$$
with probability at least $1 - 2\exp(-ct^2)$. Therefore, since $\liminf \tfrac{p}{n} > 1$, we have $\zeta_k \geq \zeta_n = \Omega_p(p^{1/2})$, as required.

If we are under the third set of assumptions, we can assume $\frac{p}{n} \to 1$, because otherwise we would be under the second set of assumptions. The empirical distribution of the nonzero singular values of $\frac{1}{n}Z^TZ$ converges to the Marchenko-Pastur density supported on $[0, 4]$ \citep{marchenko1967distribution}, which is given by
$$\frac{1}{2\pi} \sqrt{\frac{4 - x}{x}}\1\{x \in [0, 4]\}$$

Let $t = \limsup \frac{k}{n} < 1$. Then we can choose $0 < \delta < 1 -t$ and $z > 0$ such that $\P(\zeta > z) = t + \delta$, where $\zeta$ is drawn from the Marchenko-Pastur density given above. 

We have 
$$\frac{\text{\# singular values of }\frac{Z^TZ}{n} \text{ larger than } z}{n} \to t + \delta$$
so the number of singular values of $\frac{Z^TZ}{n}$ which are larger than $z$ will eventually be larger than $nt > k$, therefore $\frac{d_k^2}{n} > z > 0$ eventually, so $\zeta_k = \Omega_p(n^{1/2}) = \Omega_p(p^{1/2})$, as required. 
\end{proof}

\PCAcompatibility*
\begin{proof}
Let us denote $C=H \Gamma$ and let us have $\tilde{C} = X - \tilde{X} = U(D-\tilde{D})V^T$, where $(D - \tilde{D})$ is a $(r \times r)$-dimensional diagonal matrix (recall $r=\min(n,p)$) whose first $q$ diagonal elements correspond to the first $q$ diagonal elements of $X$ and the remaining $r-q$ elements are zero. $\tilde{C}$ corresponds to the first $q$ principal components and is a commonly used estimator of $C$ in the factor analysis literature.

From Theorem 3 in \citet{bai2003inferential} we have for $q$ fixed that
$$\min(\sqrt{n}, \sqrt{p})(\tilde{C}_{ij} - C_{ij}) = N(0, \sigma^2_{ij}),$$
where $\sigma_{ij}^2 < B$ is bounded from above. By using the standard union bound argument as in Lemma \ref{proof:tau}, we have that 
$$\norm{\tilde{C} - C}_\infty = \O_p\left(\sqrt{\frac{\log p}{\min(n, p)}}\right)$$

Therefore, since we have $X = C - \tilde{C} + E$, we obtain by triangle inequality, the fact that $\norm{x}_2 \leq \sqrt{p}\norm{x}_1$ for $x \in \R^p$ and the Hölder's inequality that
\begin{align*}
\frac{1}{\sqrt{n}}\norm{X\alpha}_2 
&\geq \frac{1}{\sqrt{n}}\norm{E\alpha}_2 - \frac{1}{\sqrt{n}}\norm{(\tilde{C} - C)\alpha}_2 \\
&\geq \frac{1}{\sqrt{n}}\norm{E\alpha}_2 - \frac{1}{\sqrt{n}}\sqrt{p}\norm{\tilde{C} - C}_\infty\norm{\alpha}_1 \\
&\geq \frac{1}{\sqrt{n}}\norm{E\alpha}_2 - \frac{1}{\sqrt{n}}\sqrt{p}\norm{\tilde{C} - C}_\infty\left(6\norm{\alpha_S}_1\right).
\end{align*}

Therefore, from the definition of the compatibility constant, we now have
$$\phi_{\tilde{\Sigma}} \geq \phi_{\tfrac{1}{n}E^TE} - \frac{6\sqrt{ps}}{\sqrt{n}}\norm{\tilde{C} - C}_\infty \to \phi_{\tfrac{1}{n}E^TE},$$
in probability, since
$$\sqrt{ps}\norm{\tilde{C} - C}_\infty = \O_p\left(\sqrt{\frac{ps}{n}}\sqrt{\frac{\log p}{\min(n, p)}}\right)$$ 
converges to $0$ in probability.

Finally, one can obtain that $\phi^2_{\tfrac{1}{n}E^TE} = \Omega_p(\lambda_{\min}(\Sigma_E))$ from the standard argument as in \cite{buhlmann2011statistics}, i.e. by using that $\tfrac{1}{n}E^TE$ concentrates around its expectation $\Sigma_E$.
\end{proof}


\compatibilityineq*
\begin{proof}
We have 
$$\alpha^T\tilde{\Sigma}\alpha = \sum_{i \leq r} \tilde{d}_i^2(V_i^T\alpha)^2 = \sum_{i \leq r} (\tilde{d}_{(i)}^2 - \tilde{d}_{(i+1)}^2)\sum_{j \leq i}(V_{(j)}^T\alpha)^2$$
where we define $\tilde{d}_{r+1} = 0$ for convenience.
Now using the fact that the infimum of the sum is not smaller than the sum of the infimums, we get
$$\phi_{\tilde{\Sigma}}^2 \geq \sum_{i \leq r} \frac{1}{n}(\tilde{d}_{(i)}^2 - \tilde{d}_{(i+1)}^2) \phi_{M_i}^2 = \sum_{i \leq r} \frac{1}{n}\tilde{d}_{(i)}^2(\phi_{M_i}^2 - \phi_{M_{i-1}}^2)$$
where $M_0$ is defined as the null matrix for convenience.
Let us now fix $k \leq r$. By using that the sequence $\tilde{d}_{(i)}$ is decreasing, we have 
$$\sum_{i \leq r} \frac{1}{n}\tilde{d}_{(i)}^2(\phi_{M_i}^2 - \phi_{M_{i-1}}^2) \geq \sum_{i \leq k} \frac{1}{n}\tilde{d}_{(k)}^2(\phi_{M_i}^2 - \phi_{M_{i-1}}^2) = \frac{1}{n}\tilde{d}_{(k)}^2\phi_{M_k}^2$$
which finishes the proof.
\end{proof}

\uniformV*
\begin{proof}
Let $Z \in \R^{k \times p}$ be a random matrix whose components are $Z_{ij} \iid N(0,1)$. Let $Z = U_Z D_Z V_Z^T$ be its SVD and let $\zeta_1 \geq \ldots \geq \zeta_k$ be its singular values.

Since $V$ is independent of $D$, $[V_{(1)}, \ldots, V_{(k)}]$ is uniform on the Stiefel manifold as well. This matrix has the same distribution as the matrix $V_Z$ and thus $M_k$ has same distribution as $V_ZV_Z^T$

From the Lasso theory (\cite{buhlmann2011statistics}) we know that $\phi^2_{\frac{1}{k}Z^TZ} \geq 1/2$ with high probability. On the other hand we have $\phi^2_{\frac{1}{k}Z^TZ} \leq \frac{1}{k}\zeta_1^2 \phi_{V_ZV_Z^T}^2$.

From Corollary 5.35. of \cite{vershynin2016high} we know 
$$\zeta_1 \leq \sqrt{p} + \sqrt{k} + C\sqrt{\log p}$$
with probability at least $1-2p^{-C^2/2}$, for any $C > 0$. This implies that $\zeta_1 = \O_p(\sqrt{p})$. By combining those results, we have that $\phi_{V_ZV_Z^T}^2 = \Omega_p\left(\frac{k}{p}\right) = \Omega_p\left(\frac{n}{p}\right)$, which finishes the proof.
\end{proof}

\paragraph{Acknowledgments.} The research of D. \'Cevid and P. B\"{u}hlmann was supported by the European Research Council under the Grant Agreement No 786461 (CausalStats - ERC-2017-ADG).

\vskip 0.2in
\bibliography{bibfile}

\end{document}